\documentclass{sig-alternate}
\usepackage[T1]{fontenc}
\usepackage{amsmath,amssymb}
\usepackage{units}      % gives the "nicefrac" command
\usepackage{graphicx}   % allows the inclusion of .eps figures
\usepackage{subfigure}
\usepackage{algorithm}% http://ctan.org/pkg/algorithms
\usepackage{algpseudocode}
\usepackage{xspace}
\usepackage{multirow}

\usepackage{enumerate}
\usepackage{marginnote}
\usepackage{color}
\usepackage{url}
\usepackage[usenames,dvipsnames]{xcolor}
\usepackage[linktocpage=true,colorlinks=false]{hyperref}

\newcommand{\remove}[1]{{}}  % removes a portion of text

\newcommand{\calx}{\mathcal{X}}
\newcommand{\caly}{\mathcal{Y}}
\newcommand{\calz}{\mathcal{Z}}

\newcommand{\calp}{\mathcal{P}}

\newcommand{\dx}{{d_{\scriptscriptstyle\calx}}}
\newcommand{\dy}{{d_{\scriptscriptstyle\caly}}}
\newcommand{\euclid}{{d_2}}

\newcommand{\hamming}{{d_h}}

\newcommand{\dprob}{{d_\calp}}

\newcommand{\priv}[1]{$#1$-privacy}
\newcommand{\privadj}[1]{$#1$-private}
\newcommand{\edp}[2][\epsilon]{$#1#2$-privacy}
\newcommand{\edpadj}[2][\epsilon]{$#1#2$-private}

\newcommand{\geoind}{geo-indistingui\-sha\-bility}

\newcommand{\smallsum}[1]{\textstyle{\sum_{#1}\:}}

\newcommand{\dspan}{$\delta$-spanner\xspace}

\usepackage{thmtools,thm-restate}
\newcommand{\expdist}{\textsc{ExpDist}}
\newcommand{\expdistarg}[3]{\expdist(#1, #2, #3)}
\newcommand{\sql}{\textsc{QL}}
\newcommand{\sqlarg}[3]{\sql(#1, #2, #3)}
\newcommand{\adv}{\textsc{AdvError}}
\newcommand{\advarg}[3]{\adv(#1, #2, #3)}
\newcommand{\plap}{\textsc{PL}}
\newcommand{\eeop}{\textsc{OptPriv}}

\newcommand{\giop}{\textsc{OptQL}}

\newcommand{\dq}{d_{Q}}
\newcommand{\da}{d_{A}}
\newcommand{\dpopt}[3]{$#2$-$\giop(#1, #3)$} % 1: prior, 2: diff privacy metric, 3: quality metric
\newcommand{\dpoptbare}[1]{$#1$-$\giop$} % 1: prior, 2: diff privacy metric, 3: quality metric
\newcommand{\sqlopt}[3]{\optpriv{#1}{#2}{#3}{q}}     % REMOVE (for now we call \optpriv with q)
\newcommand{\optpriv}[4]{$#4$-$\eeop(#1, #2, #3)$} % 1: prior, 2: adv error metric, 3: quality metric, 4: quality bound
\newcommand{\optprivbare}[1]{$#1$-$\eeop$} % 1: prior, 2: adv error metric, 3: quality metric, 4: quality bound
\newcommand{\uprof}{\pi}
\newcommand{\prior}{\pi}
\newcommand{\edx}{\epsilon\dx}
\newcommand{\edg}{\frac{\epsilon}{\delta}d_G}

\newtheorem{definition}{Definition}

\newtheorem{theorem}{Theorem}

\numberofauthors{2}

\author{
\alignauthor
	Nicol\'as E. Bordenabe\\
    \affaddr{INRIA and \'Ecole Polytechnique}\\
    \email{\normalsize nbordenabe@lix.polytechnique.fr}
\alignauthor
	Konstantinos Chatzikokolakis\\
    \affaddr{CNRS and \'Ecole Polytechnique}\\
    \email{\normalsize kostas@lix.polytechnique.fr}
\and
\alignauthor 
	Catuscia Palamidessi\\
    \affaddr{INRIA and \'Ecole Polytechnique}\\
    \email{\normalsize catuscia@lix.polytechnique.fr}
}

\newfont{\mycrnotice}{ptmr8t at 7pt}
\newfont{\myconfname}{ptmri8t at 7pt}
\let\crnotice\mycrnotice%
\let\confname\myconfname%

\makeatletter   
\def\@copyrightspace{\relax}
\makeatother

\begin{document}

\pagestyle{plain}

\title{Optimal Geo-Indistinguishable Mechanisms for\\Location Privacy}

\maketitle

\begin{abstract}
%With location-based services becoming increasingly more popular, serious concerns are being raised about the potential privacy breaches that the disclosure of location information may induce. 
We consider the geo-indistinguishability approach to location privacy, and the
trade-off with respect to utility. We show that, given a  desired degree of
geo-indistinguishability, it is possible  to construct a mechanism that
minimizes the service quality loss, using linear programming techniques. In
addition we show that, under certain conditions, such mechanism also provides optimal
privacy in the sense of Shokri et al. Furthermore, we propose a method to reduce the
number of constraints of the linear program from cubic to quadratic, maintaining
the privacy guarantees and without affecting significantly the utility of the
generated mechanism. This reduces considerably the time required to solve the
linear program, thus enlarging significantly the location sets for which the
optimal mechanisms can  be computed.
\end{abstract}



% !TEX root = geoopt.tex

\section{Introduction}
%In recent years, the increasing use of smart mobile devices equipped with GPS chips and wireless data connections has led to a massive development and utilization of of systems that record and process location data, generally referred to as ``location-based systems''. In particular, Location Based Services (LBSs), in which a user obtains, typically in real-time, a service related to his current location, have grown dramatically in popularity and sophistication. Examples of LBSs include mapping applications (e.g., Google Maps), Points of Interest (POI) retrieval (e.g., AroundMe), coupon/discount providers (e.g., GroupOn), GPS navigation (e.g., TomTom), and location-aware social networks (e.g., Foursquare).

While location-based systems (LBSs) have demonstrated to provide enormous benefits to individuals and society, these benefits come at the cost of users' privacy: as discussed
%, for instance, 
in  \cite{Freudiger:11:FC,Golle:09:PC,Krumm:07:PC}, location data can be easily linked to a variety of other information about an individual, and expose sensitive aspects of her private life such as her home address, her political views, her religious practices, etc..  
There is, therefore, a growing interest  in the development of  location-privacy protection mechanisms (LPPMs),  that allow  to use LBSs while providing sufficient privacy guarantees for the user. Most of the approaches in the literature are based on perturbing  the information reported to the  LBS provider, so to prevent the disclosure of the user's  location  \cite{Beresford:03:Perv,Chow:09:WPES,Freudiger:09:PETS,Hoh:07:CCS,Shokri:12:CCS,Andres:13:CCS}.  

Clearly, the perturbation of the information sent to the LBS provider  leads to
a degradation of the quality of service, and consequently there is a trade-off
between the level of privacy that the user wishes to guarantee and the  service
quality loss (QL) that she will have to accept. The study of this trade-off, and
the  design of mechanisms which optimize it, is an important research direction
started with the seminal paper of Shroki et al. \cite{Shokri:11:SP}. 

Obviously, any such study must be based on meaningful  notions of privacy and
of quality loss. The  authors of \cite{Shokri:11:SP} consider the privacy
threats  deriving from a Bayesian adversary. More specifically, they assume that
the adversary knows the prior probability  distribution on the user's possible
locations, and they quantify privacy as the expected error, namely the expected
distance between the true location and the best guess of the adversary once she
knows the location  reported to the LBS. We refer to this quantity as $\adv$.
The adversary's guess takes into account the
information already in her possession (the prior probability), and it is by
definition more accurate, in average, than the reported location.  We also say
that the adversary may \emph{remap} the reported location.

The notion of quality loss adopted in \cite{Shokri:12:CCS} is also defined  in terms of the expected distance between the real location and the reported location, with the important difference that the LBS is  not assumed to know the user's prior distribution
(the LBS is not tuned for any specific user), and consequently it does not apply any remapping. 
Note that the notion of distance used for expressing  QL does not need to be the same as the one used to measure location privacy. 
When these two notions coincide, then  QL is always greater than or equal to the location privacy, due to the fact that the adversary can make use of the prior information to her advantage. 
The optimal mechanism of \cite{Shokri:12:CCS} is defined as the one which maximizes  privacy for a given QL threshold, 
and since these measures are linear functions of the noise (characterized by the conditional probabilities of each   reported location given a true location), such mechanism can be computed by solving a linear optimization problem. 

In this paper, we consider the geo-indistinguishability fra\-mework of
\cite{Andres:13:CCS}, a notion of location privacy based on differential
privacy~\cite{Dwork:06:TCC}, and more precisely, on its extension to arbitrary metrics proposed in \cite{Chatzikokolakis:13:PETS}. Intuitively, a mechanism provides geo-indistinguishability if two locations that are geographically close have similar probabilities to generate a certain reported location. Equivalently,  the reported location will not increase by much the adversary's  chance to distinguish the true location among the nearby ones.
Note that this notion protects the accuracy of the location:  the adversary is allowed to distinguish locations which are far away.  
It is important to note that the property of geo-indistinguishability does not depend on the prior.  
This is a  feature inherited from differential privacy,  which makes the
mechanism robust with respect to composition of attacks in the same sense as differential privacy.%\footnote{This does not mean that the prior knowledge is not harmful for privacy: it is harmful, for both geo-indistinguishability and differential privacy, as it is for any obfuscation mechanism. But the compositionality guarantees that if the prior knowledge is acquired only via differentially private mechanisms, then the loss of privacy is gradual and under control.} 

We study the problem of optimizing the trade-off between geo-indistinguishability and quality of service. 
More precisely, given a certain threshold on the degree of geo-indistin-guishability, and a prior, we aim at obtaining the mechanism $K$ which minimizes 
QL.  Thanks to the fact that the property of respecting the geo-indistinguishability threshold can be expressed 
by linear constraints, we can reduce the problem of producing such a  $K$ to a linear optimization problem, which can then be solved by using standard techniques of linear programming. 

It should be remarked that our approach is, in a sense, dual wrt the one of
\cite{Shokri:12:CCS}. The latter fixes a bound on  QL and optimizes the location
privacy. Here, on the contrary, we fix a bound on the location privacy and then
optimize  QL. Another important difference is that in \cite{Shokri:12:CCS} the
privacy degree of the optimal mechanism, measured by $\adv$, is  guaranteed for a specific prior
only, while in our approach the privacy guarantee of the optimal mechanism is in
terms of geo-indistinguihability, which does not depend on the prior. In our
opinion, this is an important feature of the present approach, as it is
difficult to control the prior knowledge of the adversary.  Consider, for
instance,  a  user for which the optimal mechanism has been computed with
respect to his average day (and consequent prior $\pi$), and who has very
different habits in the morning and in the afternoon. By simply taking into
account the time of the day, the adversary gains some additional knowledge that
determines  a different prior, and the privacy guarantees of the optimal
mechanism of \cite{Shokri:12:CCS}  can be severely violated  when the adversary
uses a prior different from $\pi$. 

However, when the notion of distance used to
measure the QL coincides with that used for expressing the degree of
privacy according to $\adv$, then, somewhat surprisingly, our optimal mechanism
$K$ turns out to be also optimal in terms of
$\adv$, in a sense getting the best of both approaches. 
Intuitively, this is due to the fact that the property of  geo-indistinguishability is not affected by remapping. 
Hence, the expected error of the adversary must coincide with  QL, i.e., the adversary cannot gain anything by any remapping $H$, or otherwise  
$KH$ would be still geo-indistinguishable and provide a better QL. Since privacy
coincides with the QL, it must also be optimal.
In conclusion, we obtain a geo-indistinguishable $K$ with minimum  QL and
maximum degree of privacy (for that QL). 

Note that the optimal mechanisms are not unique, and ours does not usually coincide with the one produced by the algorithm of \cite{Shokri:12:CCS}. 
In particular the one of  \cite{Shokri:12:CCS} in general does not  
 provide  geo-indistinguishability, while ours does, by design. 
The robustness of the  geo-indistinguishability property seems to affect favorably also other notions of privacy: 
We have evaluated the two mechanisms with the privacy definition of
\cite{Shokri:12:CCS} on two real datasets, and we have observed that,
while the mechanism of  \cite{Shokri:12:CCS} by definition offers the best privacy on the prior for which it is computed, 
ours can perform significantly better when we consider different priors. 

We now turn our attention to efficiency concerns.  
Since the optimal mechanism is obtained by solving a linear optimization problem, the efficiency depends crucially on the number of constraints  used to express geo-indistinguishability. We note that this number is, in general, cubic with respect to the amount of locations considered. 
We show that we are able to reduce this number from cubic to quadratic, using an approximation technique  
based on constructing a suitable spanning graph of the set of locations. 
The idea is that, instead of considering the geo-indistinguishability constraints for every pair of locations, we only consider those for every edge in the spanning graph. 
We also show, based on experimental results, that for a reasonably good approximation our approach offers an improvement in running time with respect to method of Shokri et al.
We must note however that the mechanism obtained this way is no longer optimal with respect to the original metric, but only with respect to the metric induced by the graph, 
and therefore the $\sql$ of the mechanism might be higher, although our experiments also show that this increase is not significant.

Note that in this paper we focus on the case of \emph{sporadic} location disclosure, 
%that is, we ignore attacks based on the correlation between  consecutive locations reported by a user.  
that is, we assume that there is enough time between consecutive locations reported by the user, and therefore they can be considered independent.   
Geo-indistinguishability can be applied also in case of correlation between consecutive points,  but additional care must be taken to avoid 
the degradation of privacy, that could be significant when the number of consecutive locations is high.
The problem of correlation is orthogonal to to the goals of this paper. 
We refer to \cite{Chatzikokolakis:14:PETS} for a study of this problem.

\paragraph{Contribution} The main contributions of this paper are the following:
\begin{itemize}
\item We present a method based on  linear optimization   to generate a mechanism that is geo-indistinguishable and achieves optimal utility. 
Furthermore when the notions of distance used for QL coincide with  that used for geo-indistinguishability, then the mechanism is also optimal with respect to the expected error of the adversary.
\item We evaluate our approach under different priors (generated from real
	traces of two widely used datasets), and show that it outperforms the other mechanisms considered.
\item We propose an approximation  technique, based on spanning graphs, that can be used to reduce the number of constraints of the optimization problem and still obtain a geo-indistinguishable mechanism.
\item We measure the impact of the approximation on the utility and the number of constraints, and analyze the running time of the whole method, obtaining favorable results.
\end{itemize}

\paragraph{Plan of the paper}
The rest of the paper is organized as follows.
Next section recalls some preliminary notions. 
In Section~\ref{mechanism} we illustrate our method to 
produce a geo-indistinguishable and optimal mechanism
as the solution of a linear optimization problem, and we 
propose a technique to reduce the number of constraints 
used in the problem. 
In Section \ref{sec:evaluation} we evaluate our mechanism
with respect to  other ones in the literature.
Finally, in Section \ref{related}, we discuss 
 related work and   conclude.

\bigskip
This paper is the report version of a work that appeared in the proceedings of the 21st ACM Conference on Computer Security. Scottsdale, Arizona, USA, Nov. 2014 (CCS'14).

% !TEX root = geoopt.tex

\section{Preliminaries}

%- I should pick a name for our mechanism, and also for Reza's
%- 

%Differential privacy in databases:
%	- definition on databases
%	- generalization (referring to the PETS paper)
%Geo indistinguishability
%	- definition
%	- talk about the planar laplace mechanism
%Reza's mechanism
%	- mention that the map is divided in regions
%	- introduce the user profiles
%	- introduce metrics for privacy and utility
%	- linear program? (i don't think it necessary)
%	- name the mechanism	

%In this section, we recall some notions from the literature, which
%will be used in the following sections.

\subsection{Location obfuscation, quality loss and adversary's error}
\label{sec:rezas-mech}

A common way of achieving location privacy is to apply a \emph{location
obfuscation} mechanism, that is a probabilistic function
$K:\calx\to\calp(\calx)$ where $\calx$ is the set of possible locations,
and $\calp(\calx)$ denotes the set of probability distributions over $\calx$.
$K$ takes a location $x$ as input, and produces a \emph{reported location}
$z$ which is communicated to the service provider. In this paper we generally
consider $\calx$ to be finite, in which case $K$ can be represented by a
stochastic matrix, where $k_{xz}$ is the probability to report $z$ from
location $x$.

A prior distribution $\prior \in\calp(\calx)$ on the set of locations
can be viewed either as modelling the behaviour of the user (the \emph{user
profile}), or as capturing the adversary's \emph{side information} about the user.
Given a prior $\pi$ and a metric $d$ on $\calx$, the expected distance between the real
and the reported location is:
\[
	\expdistarg{K}{\prior}{d} = \smallsum{x, z} \pi_x k_{xz} d(x, z)
\]

From the user's point of view, we want to quantify the service \emph{quality
loss (QL)} produced by the mechanism $K$. Given a \emph{quality metric} $\dq$ on
locations, such that $\dq(x,z)$ measures how much the quality decreases by reporting
$z$ when the real location is $x$ (the Euclidean metric $\euclid$ being a
typical choice), we can naturally define the quality loss as the expected
distance between the real and the reported location, that is 
$
	\sqlarg{K}{\prior}{\dq} = \expdistarg{K}{\prior}{\dq}
$.
The QL can also be viewed as the (inverse of the) utility of the
mechanism.

Similarly, we want to quantify the \emph{privacy} provided by $K$. A natural
approach, introduced in \cite{Shokri:11:SP} is to consider a Bayesian adversary
with some prior information $\pi$, trying to remap $z$ back to a guessed
location $\hat{x}$. A remapping strategy can be modelled by a stochastic matrix
$H$, where $h_{z\hat{x}}$ is the probability to map $z$ to $\hat{x}$. Then the
privacy of the mechanism can be defined as the expected error of an adversary
under the best possible remapping:
\[
	\advarg{K}{\pi}{\da} = \min_H \expdistarg{KH}{\prior}{\da}
\]
Note that the composition $KH$ of $K$ and $H$ is itself a mechanism.
Similarly to $\dq$, the metric $\da(x,\hat{x})$ captures the adversary's
loss when he guesses $\hat{x}$ while the real location is $x$. Note that
$\dq$ and $\da$ can be different, but the canonical choice is to use the Euclidean
distance for both.

A natural question, then, is to construct a mechanism that achieves \emph{optimal
privacy}, given a \emph{QL constraint}.
\begin{definition}
\label{def:sql-opt}
	Given a prior $\prior$, a quality
	metric $\dq$, a quality bound $q$ and an adversary metric $\da$, a mechanism $K$ is
	\optpriv{\prior}{\da}{\dq}{q} iff
	\begin{enumerate}
		\item $\sqlarg{K}{\prior}{\dq} \le q$, and
		\item for all mechanisms $K'$, $\sqlarg{K'}{\prior}{\dq} \le q$ implies
		$
		\advarg{K'}{\prior}{\da} \le \advarg{K}{\prior}{\da}
		$
	\end{enumerate}
\end{definition}
In other words, a \optprivbare{q} mechanism 
provides the best privacy (expressed in terms of
$\adv$) among all mechanisms with $\sql$ at most $q$.
This problem was studied in \cite{Shokri:12:CCS}, providing a method
to construct such a mechanism for any $q,\pi,\da,\dq$, by solving a properly
constructed linear program.

\subsection{Differential privacy}

Differential privacy was originally introduced in the context of statistical
databases, requiring that a query should produce similar results when applied
to \emph{adjacent} databases, i.e. those differing by a single row.
The notion of adjacency is related to the Hamming metric $\hamming(x,x')$ defined
as the number of rows in which $x,x'$ differ. Differential privacy requires that
the greater the hamming distance between $x,x'$ is, the more distinguishable
they are allowed to be.

This concept can be naturally extended to any set of secrets $\calx$, equipped
with a metric $\dx$ \cite{Reed:10:ICFP,Chatzikokolakis:13:PETS}.
The distance $\dx(x,x')$ expresses the \emph{distinguishability level} between
$x$ and $x'$: if the distance is small then the secrets should remain
indistinguishable, while secrets far away from each other are allowed to
be distinguished by the adversary. The metric should be chosen depending on
the application at hand and the semantics of the privacy notion that we try to
achieve.

Following the notation of \cite{Chatzikokolakis:13:PETS}, a mechanism is a
probabilistic function $K:\calx \to\calp(\calz)$, where $\calz$ is a set of
\emph{reported values} (assumed finite for the purposes of this paper). The
similarity between probability distributions can be
measured by the multiplicative distance $\dprob$ defined as $ d_{\cal
P}(\mu_1,\mu_2) = \sup_{z \in \calz} |\ln \frac{\mu_1(z)}{\mu_2(z)}| $
with $|\ln \frac{\mu_1(z)}{\mu_2(z)}|=0$ if both $\mu_1(z),\mu_2(z)$ are zero
and $\infty$ if only one of them is zero. In other words,
$\dprob(\mu_1,\mu_2)$ is small iff $\mu_1,\mu_2$ assign similar probabilities to each
value $z$.

The generalized variant of differential privacy under the metric $\dx$, called
\priv{\dx}, is defined as follows:
\begin{definition}%[\priv{\dx}]
  A mechanism $K : \calx\rightarrow \calp(\calz)$ satisfies
  \priv{\dx} iff:
  \[
	  \dprob(K(x),K(x')) \leq \dx(x,x')
		\qquad \forall x,x' \in \calx
  \]
\end{definition}
or equivalently
$
	 K(x)(z) \leq e^{\dx (x,x')} K(x')(z)
$
for all $x,x'\in\calx,z\in\calz$.
A privacy parameter $\epsilon$ can also be introduced by scaling
the metric $\dx$ (note that $\epsilon\dx$ is itself a metric).

Differential privacy can then be expressed as \priv{\epsilon\hamming}. Moreover,
different metrics give rise to various privacy notions of interest;
several examples are given in \cite{Chatzikokolakis:13:PETS}.

\subsection{Geo-indistinguishability}
\label{sec:geoind}

In the context of location based systems the secrets $\calx$ are locations, and
we can obtain a useful notion of location privacy by naturally using
the Euclidean distance $\euclid$, scaled by a security parameter $\epsilon$.
The resulting notion of \priv{\epsilon\euclid}, called $\epsilon$-geo-indistinguishability
in \cite{Andres:13:CCS}, requires that a location obfuscation mechanism should
produce similar results when applied to locations that are geographically
close. This prevents the service provider from inferring the user's location
with accuracy, while allowing him to get approximate information required to
provide the service. Following the spirit of differential privacy, this
definition is independent from the prior information of the adversary.

A characterization of geo-indistinguishability from \cite{Andres:13:CCS}
provides further intuition about this notion. The characterization compares the
adversary's conclusions (a posterior distribution) to his initial knowledge (a
prior distribution). Since some information is supposed to be revealed (i.e. the
provider will learn that the user is somewhere around Paris), we cannot expect
the two distributions to coincide. However, geo-indistinguishability implies
that an \emph{informed adversary} who already knows that the user is located
within a small area $N$, cannot improve his initial knowledge and locate the
user with higher accuracy. More details, together with a second characterization
can be found in \cite{Andres:13:CCS}.

Note that geo-indistinguishability does not guarantee a small leakage under any
prior; in fact no obfuscation mechanism can ensure this while offering some
utility. Consider, for instance, an adversary who knows that the user is located
at some airport, but not which one. Unless the noise is huge, reporting an
obfuscated location will allow the exact location to be inferred, but
this is unavoidable.\footnote{This example is the counterpart of the well-known Terry
Gross example from \cite{Dwork:06:TCC}}.

Considering the mechanism, \cite{Andres:13:CCS} shows that
\geoind{} can be achieved by adding noise to the user's location
drawn from a 2-dimensional Laplace distribution. This can be easily done in
polar coordinates by selecting and angle uniformly and a radius from a Gamma
distribution. If a restricted set of reported locations is allowed, then the
location produced by the mechanism can be mapped back to the closest among the
allowed ones.

Although the Laplace mechanism provides an easy and practical way of achieving
geo-indistinguishability, independently from any user profile, its utility
is not always optimal. In the next section we show that by tailoring a mechanism
to a prior corresponding to a specific user profile, we can achieve better
utility for that prior, while still satisfying geo-indistinguishability, i.e. a
privacy guarantee independent from the prior. The evaluation results in
Section~\ref{sec:evaluation} show that the optimal mechanism can provide
substantial improvements compared to the Laplace mechanism.

% !TEX root = geoopt.tex

\section{Geo-indistinguishable \\mechanisms of optimal utility}\label{mechanism}

%In this section, we describe the method to obtain a \edpadj{\dx} mechanism with optimal utility.
As discussed in the introduction, we aim at obtaining a mechanism that optimizes the tradeoff 
between privacy (in terms of geo-indistinguishability) and quality loss (in terms the metric $\sql$). 
Our main goal is, given a set of locations $\calx$ with a privacy metric $\dx$ (typically the
Euclidean distance), a privacy level $\epsilon$, a user profile $\uprof$ and a
quality metric $\dq$, to find an \edpadj{\dx} mechanism such that its $\sql$ is
as small as possible.
%In this section we describe a method to obtain such a mechanism.

We start by describing a set of linear constraints that enforce \edp{\dx}, which
allows to obtain an optimal mechanism as a
linear optimization problem.
%These constraints will be used to construct an optimization problem, with the
%mechanism as result. 
However, the number of constraints can be large, making the approach
computationally demanding as the number of locations increases.
As a consequence, we propose an approximate solution that replaces $\dx$ with the metric induced by
a spanning graph.
%However, the mechanism obtained this way might no longer be
%the one that minimizes the $\sql$. In fact, how close its $\sql$ is from that of
%the optimal one depends on the approximation factor of the spanner. 
%%Since computing a ``good'' spanner might be hard, 
We discuss a greedy algorithm to calculate the spanning graph and analyze its
running time.  We also show that, if the quality and adversary metrics
coincide, then the constructed (exact or approximate) mechanisms
also provide optimal privacy in terms of $\adv$. Finally, we discuss some
practical considerations of our approach.

%\break
\subsection{Constructing an optimal mechanism}
\label{sec:approach}

%1. problem description: talk about the division into regions, the privacy constraints, the service quality loss...
%2. describe the constraints: talk and develop the linear program
%3. claim that the resulting mechanism will both be epsilon-ge-ind and have better utility than any other mechanism with the same epsilon

%%% TALK ABOUT THE REGIONS (arbitrary shape and size), mention that the user reports a whole region as his location, give examples of what these regions may represent (arrondissements, neighborhoods, etc). Give examples of the distance function. Maybe even examples of the user profile (probabilities of being at work, at home, etc).

The constructed mechanism is assumed to have as both input and output a
predetermined finite set of locations $\calx$. For instance, $\calx$ can be
constructed by dividing the map in a finite number of regions (of arbitrary size
and shape), and selecting in $\calx$ a representative location for each region.
We also assume a prior $\uprof$ over $\calx$, representing the
probability of the user being at each location at any given time. 
%In general, the probabilities assigned by $\psi_u$ to the regions corresponding to the user's home and place of work should be considerably high, while regions corresponding to abandoned areas of the city should have negligible or zero probabilities assigned.

Given a privacy metric $\dx$ (typically the
Euclidean distance) and a privacy parameter $\epsilon$, the goal is to construct
a \edpadj{\dx} mechanism $K$ such that the
\emph{service quality loss} with respect to a quality metric $\dq$ is minimum.
This property is formally defined below:

\begin{definition}
\label{def:dx-opt}
	Given a prior $\prior$, a privacy metric $\dx$, a privacy parameter
	$\epsilon$ and a quality metric $\dq$, a mechanism $K$ is
	\dpopt{\prior}{\epsilon\dx}{\dq} iff:
	\begin{enumerate}
		\item $K$ is \privadj{\epsilon\dx}, and
		\item for all mechanisms $K'$, if $K'$ is \privadj{\epsilon\dx} then\\
		$
		\sqlarg{K}{\prior}{\dq} \leq \sqlarg{K'}{\prior}{\dq}
		$
	\end{enumerate}
\end{definition}
Note that \dpoptbare{\epsilon\dx} optimizes $\sql$ given a privacy constraint,
while \optprivbare{q} (Definition~\ref{def:sql-opt}) optimizes
privacy, given an $\sql$ constraint.

In order for $K$ to be \edpadj{\dx}
it should satisfy the following constraints:
\[
\begin{array}{l c l}
k_{xz} \leq e^{\epsilon \dx(x,x')} k_{x'z} & \quad \quad &  x, x', z \in \calx \\
\end{array}
\]
Hence, we can construct an optimal mechanism by solving a linear optimization problem,
minimizing $\sqlarg{K}{\uprof}{\dq}$ while satisfying \edp{\dx}:
%\begin{align*}
%    \textbf{choose}   \\
%			& k_{xz},\quad x,z \in \calx \\
%    \\
%    \textbf{to minimize}   \\
%			& \sum_{x, z \in \calx} \uprof(x) k_{xz} \dx(x, z)\  \tag{1}\\
%    \\
%    \textbf{subject to} \\
%			& k_{xz} \leq e^{\epsilon \dx(x,x')} k_{x'z}, \quad &x,x',z\in\calx \tag{2}\\
%			& k_{xz} \geq 0,\quad &x, z \in \calx \tag{3}\\
%			& \sum_{z \in \calx} k_{xz} = 1, \quad &x\in \calx \tag{4}
%\end{align*}
\begin{align*}
	\textbf{Minimize:}   &\quad \sum_{x, z \in \calx} \uprof_x k_{xz} \dq(x, z) \\
	\textbf{Subject to:} 
			&\quad k_{xz} \leq e^{\epsilon \dx(x,x')} k_{x'z}  &x,x',z\in\calx \\
			&\quad \sum_{z \in \calx} k_{xz} = 1  &x\in \calx \\
			&\quad k_{xz} \geq 0 & x, z \in \calx
\end{align*}

It is easy to see that the mechanism $K$ generated by the previous optimization problem is \dpopt{\uprof}{\edx}{\dq}.
%We say that the mechanism obtained this way is \emph{optimal} because, between all the mechanisms providing \edp{\dx}, it is
%the one with the smallest $\sql$ for the user profile $\uprof$. 

%We will call this mechanism $\giopt{\psi_u}{\epsilon}$ (where $\textsc{GI}$ stands for ``Geo-Indistinguishability'').

\subsection{A more efficient method using spanners}
\label{sec:mechanism-spanner}

%In the optimization problem presented before, the \edp{\dx} definition introduces the $|\calx|^3$ constraints (2) in the linear program.
In the optimization problem of the previous section, the \edp{\dx} definition introduces $|\calx|^3$ constraints in the linear program.
However, in order to be able to manage a large number of locations, we would like to reduce this amount to a number in the order of $O(|\calx|^2)$. 
One possible way to achieve this is to use the \emph{dual form} of the linear
program (shown in the appendix). The dual program has as many constraints as
the variables of the primal program (in this case $|\calx|^2$) and one variable for
each constraint in the primal program (in this case $O(|\calx|^3)$). Since the
primal linear program finds the optimal solution in a finite number of steps, it
is guaranteed by the strong duality theorem that dual program will also do so.
However, as shown in Section \ref{sec:performance}, in practice the dual
program does not offer a substantial improvement with respect to the primal one
(a possible explanation being that, although fewer in number, the constrains in
the dual program are more complex, in the sense that each one of them involves a larger
number of variables).

An alternative approach is to exploit the structure of the metric $\dx$.
So far we are not making any assumption about $\dx$, and therefore we need to specify $|\calx|$ constraints for each pair of locations $x$ and $x'$. However, it is worth noting that if the distance $\dx$ is induced by a weighted graph (i.e. the distance between each pair of locations is the weight of a minimum path in a graph), then we only need to consider $|\calx|$ constraints for each pair of locations that are \emph{adjacent in the graph}. 
An example of this is the usual definition of differential privacy: since the
adjacency relation between databases induces the Hamming distance $d_h$, we only
need to require the differential privacy constraint for each pair of databases
that are adjacent in the Hamming graph (i.e. that differ in one individual). 

It might be the case, though, that the metric $\dx$ is not induced by any graph
(other than the complete graph), and consequently the amount of constraints remains the same. In fact, this is generally the case for the Euclidean metric. Therefore, we consider the case in which $\dx$ can be \emph{approximated} by some graph-induced metric.

If $G$ is an undirected weighted graph, we denote with $d_G$ the distance function induced by $G$, i.e. $d_G(x,x')$ denotes the weight of a minimum path between the nodes $x$ and $x'$ in $G$. 
Then, if the set of nodes of $G$ is $\calx$ and the weight of its edges is given by the metric $\dx$, we can approximate $\dx$ with $d_G$.
%Now, if we consider $\calx$ as a complete weighted graph (with the weight of each edge given by $\dx$), and we then consider a subgraph $G$ of $\calx$, we can approximate the distance $\dx$ with $d_G$.
%We can approximate the distance function $\dx$ by considering $\calx$ as a complete weighted graph (with the weight of each edge given by $\dx$) and taking the distance $d_G$ induced by a subgraph $G$ of $\calx$.
In this case, we say that $G$ is a spanning graph, or a spanner \cite{Narasimhan:07:BOOK,Sack:99:BOOK}, of $\calx$.

\begin{definition}[Spanner]
A weighted graph $G = (\calx, E)$, with $E\subseteq \calx\times\calx$ and weight function $w:E\rightarrow \mathbb{R}$ is a \emph{spanner} of $\calx$ if
\[
w(x, x') = \dx(x,x') \quad \forall (x, x')\in E
\]
\end{definition}

Note that if $G$ is a spanner of $\calx$, then
\[
d_G(x,x') \geq \dx(x,x') \quad \forall x,x'\in\calx
\]
A main concept in the theory of spanners is that of dilation, also known as stretch factor:

\begin{definition}[Dilation]
Let $G = (\calx, E)$ be a spanner of $\calx$. The \emph{dilation} of $G$ is calculated as:
\[
\delta = \max_{x\neq x' \in \calx} \frac{d_G(x,x')}{\dx(x,x')}
\]
A spanner of $\calx$ with dilation $\delta$ is called a $\delta$-\emph{spanner} of $\calx$.
\end{definition}

\begin{figure}[t]
      \centering
      \includegraphics[width=\columnwidth]{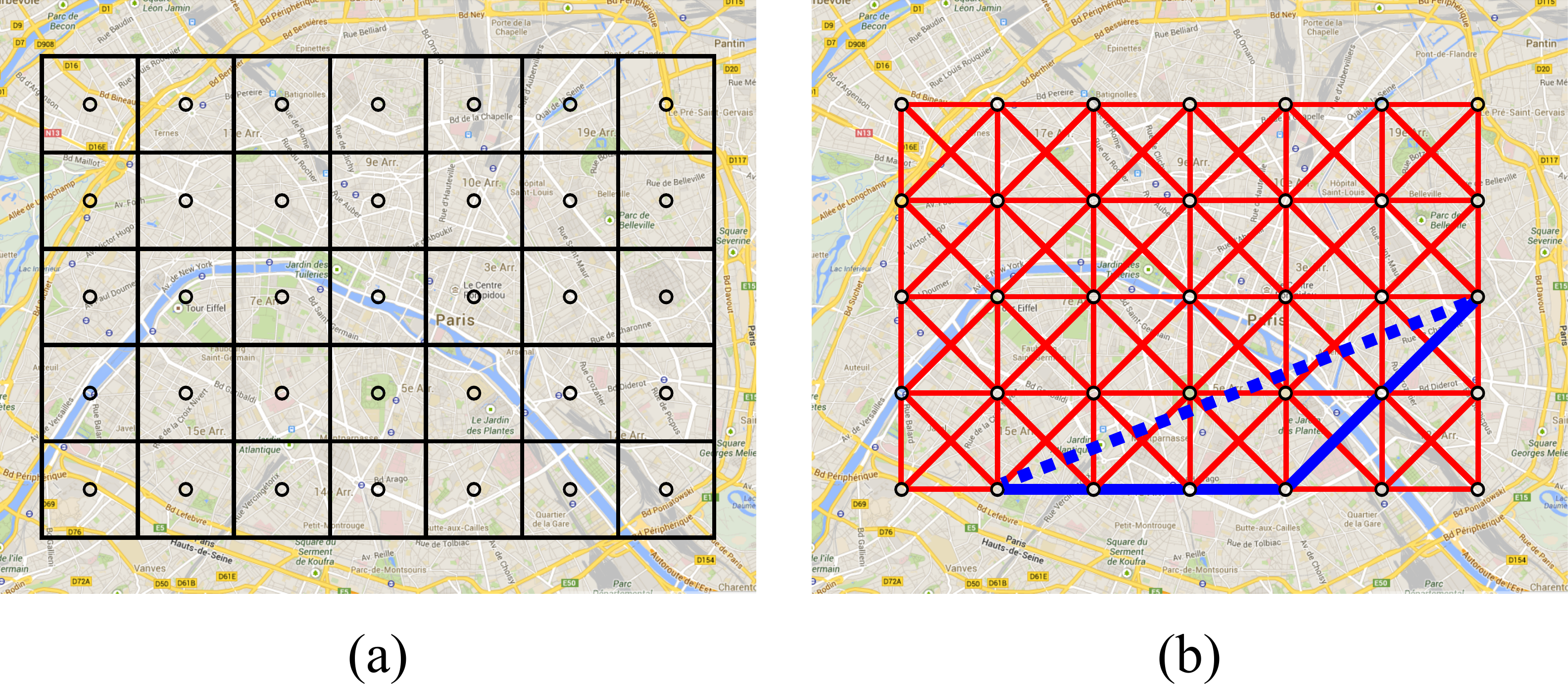}\\[1ex]
   \caption{(a) a division of the map of Paris into a $7\times 5$ square grid. The set of locations $\calx$ contains the centers of the regions. (b) A spanner of $\calx$ with dilation $\delta = 1.08$. }
   \label{fig:chart-example}
 \end{figure}

Informally, a $\delta$-spanner of $\calx$ can be considered an approximation of the metric $\dx$ in which
distances between nodes are ``stretched'' by a factor of at most $\delta$. Spanners are generally used to approximate distances in a geographic network without considering the individual distances between each pair of nodes. An example of a spanner for a grid in the map can be seen in Figure \ref{fig:chart-example}. 

If $G$ is a $\delta$-spanner of $\calx$, then it holds that
\[
d_G(x,x') \leq \delta \dx(x,x') \quad \forall x, x' \in \calx
\]
which leads to the following proposition: 

\begin{restatable}{proposition}{propprivacyimpl}
\label{prop:privacy-impl}
Let $\calx$ be a set of locations with metric $\dx$, and let $G$ be a $\delta$-spanner of $\calx$.
If a mechanism $K$ for $\calx$ is \privadj{\frac{\epsilon}{\delta}d_G}, then $K$ is \edpadj{\dx}.
\end{restatable}

We can then propose a new optimization problem to obtain a \edpadj{\dx} mechanism. 
%If $G=(\calx, E)$ is a \dspan of $\calx$, instead of specifying the constraints corresponding to \edp{\dx}, we require only the constraints needed by \priv{\edg}, that is, $|\calx|$ constraints for each edge of $G$:
If $G=(\calx, E)$ is a \dspan of $\calx$, we require not the constraints corresponding to \edp{\dx}, but those corresponding to \priv{\edg} instead, that is, $|\calx|$ constraints for each edge of $G$:
%\begin{align*}
%    \textbf{choose}   \\
%			& k_{xz}, \quad x, z \in \calx \\
%    \\
%    \textbf{to minimize}   \\
%			& \sum_{x, z \in \calx} \uprof(x) k_{xz} \dq(x, z)\  \\
%    \\
%    \textbf{subject to} \\
%			& k_{xz} \leq e^{\frac{\epsilon}{\delta} d_G(x,x')} k_{x'z}, \quad z \in \calx, (x,x') \in E \tag{7}\\
%			& k_{xz} \geq 0, \quad x, z \in \calx \\
%			& \sum_{x \in \calx} k_{xz} = 1, \quad x\in \calx
%\end{align*}
\begin{align*}
	\textbf{Minimize:}   &\quad \sum_{x, z \in \calx} \uprof_x k_{xz} \dq(x, z) \\
	\textbf{Subject to:} 
			&\quad k_{xz} \leq e^{\frac{\epsilon}{\delta} d_G(x,x')} k_{x'z} & z \in \calx, (x,x') \in E\\
			&\quad \sum_{x \in \calx} k_{xz} = 1&  x\in \calx\\
			&\quad k_{xz} \geq 0 & x, z \in \calx
\end{align*}

Since the resulting mechanism is \privadj{\frac{\epsilon}{\delta}d_G}, by Proposition \ref{prop:privacy-impl} it must also be \edpadj{\dx}. However, the number of constraints in induced by \priv{\frac{\epsilon}{\delta}d_G} is now $|E||\calx|$. 
Moreover, as discussed in the next section, for any $\delta > 1$ there is an
algorithm that generates a \dspan with $O(\frac{|\calx|}{\delta-1})$ edges,
which means that, fixing $\delta$, the total number of constraints of the linear program is  $O(|\calx|^2)$. 
%This means that, if the number of edges of the spanner is in the order of $O(|\calx|)$, then the total number of constraints of the linear program will be in the order of $O(|\calx|^2)$. 

It is worth noting that although \priv{\epsilon\dx} is guaranteed, optimality
is lost:
the obtained mechanism is \dpopt{\uprof}{\edg}{\dq} but not necessarily
\dpopt{\uprof}{\edx}{\dq}, since the set of \privadj{\edg} mechanisms is a
subset of the set of \edpadj{\dx} mechanisms.  The $\sql$ of the obtained
mechanism will now depend on the dilation $\delta$ of the spanner: the smaller
$\delta$ is, the closer the QL of the mechanism will be from the optimal one. However, if $\delta$ is too small then the number of edges of the spanner will be large, and therefore the number of constraints in the linear program will increase. In fact, when $\delta=1$ the mechanism obtained is also \dpopt{\uprof}{\edx}{\dq} (since $d_G$ and $\dx$ coincide), but the amount of constraints is in general $O(|\calx|^3)$. In consequence, there is a tradeoff between the accuracy of the approximation and the number of constraints in linear program.

%We will call $\gioptd{\delta}{\psi_u}{\epsilon}$ to the mechanism obtained as a result of the optimization problem for \priv{\frac{\epsilon}{\delta}d_G}, with respect to user profile $\psi_u$.

\subsection{An algorithm to construct a \dspan}
\label{sec:algorithm}

The previous approach requires to compute a spanner for $\calx$. Moreover, given a dilation factor $\delta$, we are interested in generating a \dspan with a reasonably small number of edges. In this section we describe a simple greedy algorithm to get a \dspan of $\calx$, presented in \cite{Narasimhan:07:BOOK}. This procedure (described in Algorithm \ref{alg:spanner}) is a generalization of 
Kruskal's minimum spanning tree algorithm. 

\begin{algorithm}
\caption{Algorithm to get a \dspan of $\calx$}\label{alg:spanner}
\begin{algorithmic}[1]
\Procedure{GetSpanner}{$\calx, \dx, \delta$}\label{alg:get-spanner:proc}
	\State $E := \emptyset$\label{alg:get-spanner:init-start}
	\State $G := (\calx, E)$\label{alg:get-spanner:init-end}
%	\ForAll{$(x,x')\in \textsc{Sort}(\calx\times\calx)$}\label{alg:get-spanner:mainloop-start}
	\ForAll{$(x,x')\in (\calx\times\calx)$}\label{alg:get-spanner:mainloop-start}
		\Comment{taken in increasing order wrt $\dx$}
		\If{$d_G(x,x') > \delta \dx(x,x')$}
			\State $E := E \cup \{(x,x')\}$
		\EndIf
	\EndFor\label{alg:get-spanner:mainloop-end}
	\State \Return $G$
\EndProcedure
\end{algorithmic}
\end{algorithm}

The idea of the algorithm is the following: we start with a spanner with an empty set of edges (lines \ref{alg:get-spanner:init-start}-\ref{alg:get-spanner:init-end}). In the main loop we consider all possible edges (that is, all pairs of locations) in \emph{increasing order} with respect to the distance function $\dx$ (lines \ref{alg:get-spanner:mainloop-start}-\ref{alg:get-spanner:mainloop-end}), and if the weight of a minimum path between the two corresponding locations in the current graph is bigger than $\delta$ times the distance between them, we add the edge to the spanner. By construction, at the end of the procedure, graph $G$ is a $\delta$-spanner of $\calx$.

A crucial result presented in \cite{Narasimhan:07:BOOK} is that,  in the case where $\calx$ is a set of points in the Euclidean plane, the degree of each node in the generated spanner only depends on the dilation factor:

\begin{theorem}
Let $\delta > 1$. If $G$ is a \dspan for $\calx \subseteq \mathbb{R}^2$,  with the Euclidean distance $d_2$ as metric, then the degree of each node in the spanner constructed by Algorithm \ref{alg:spanner} is  $O(\frac{1}{\delta - 1})$. 
\end{theorem}

%One interesting result, also presented in \cite{Narasimhan:07:BOOK} is that, in the case in which $\calx \subseteq \mathbb{R}^2$ and $\dx$ is the Euclidean metric, then for $\delta > 1$ the degree of each node in the \dspan constructed by Algorithm \ref{alg:spanner} is a number in the order of $O(\frac{1}{\delta - 1})$. 
%And since the goal is to generate a \emph{sparse} spanner (a spanner with a linear amount of edges), then this will allow us to estimate 
This result is useful to estimate the total number of edges in the spanner, since our goal is to generate a \emph{sparse} spanner, i.e. a spanner with O($|\calx|$) edges.

%It is known from \cite{Narasimhan:07:BOOK} that, if the complete graph induced by $\calx$ and the metric $\dx$ is isomorphic to a geometric graph in the 2-dimensional plane, then the degree of each node in the resulting spanner is $O(\frac{1}{\delta - 1})$. This is useful to estimate the total number of edges in the spanner, since in general we would like to keep it in $O(|\calx|)$.

Considering the running time of the algorithm, since the main loop requires all pair of regions to be sorted increasingly by distance, we need to perform this sorting before the loop. This step takes $O(|\calx|^2 \log |\calx|)$. The main loop performs a minimum-path calculation in each step, with $|\calx|^2$ total steps. If we use, for instance, Dijkstra's algorithm, each of these operations can be done in $O(|E| + |\calx| \log |\calx|)$. If we select $\delta$ so that the final amount of edges in the spanner is linear, i.e. $|E| = O(|\calx|)$, we can conclude that the total running time of the main loop is $O(|\calx|^3 \log |\calx|)$. This turns out to be also the complexity of the whole algorithm. 
%Although this could seem to be a high complexity, we should recall that the final algorithm should also solve a lineal program with $O(|\calr|)$ constraints and $|\calr|^2$ variables. 

%Although this algorithm finds a $\delta$-spanner, it does not find the \emph{best} one, that is, the $\delta$-spanner with the least amount of edges. It turns our that finding the best $\delta$-spanner for a 2-dimensional geometric graph is {\sc NP-Hard} (\cite{Klein:06:GD}). This means that finding it for an arbitrary graph is, at least, {\sc NP-Hard} as well.

A common problem in the theory of spanners is the following: given a set of points $\calx \subseteq \mathbb{R}^2$ and a maximum amount of edges $m$, the goal is to find the spanner with \emph{minimum} dilation with at most $m$ edges. This has been proven to be NP-Hard (\cite{Klein:06:GD}). In our case,  we are interested in the analog of this problem: given a maximum tolerable dilation factor $\delta$, we want to find a \dspan with minimum amount of edges. However, we can see that the first problem can be expressed in terms of the second (for instance, with a binary search on the dilation factor), which means that the second problems must be at least NP-Hard as well.

\subsection{$\adv$ of the obtained mechanism}
\label{sec:mech-ee}

As discussed in \ref{sec:rezas-mech}, the privacy of a location obfuscation
mechanism can be expressed in terms of $\adv$ for an adversary metric $\da$. In
\cite{Shokri:12:CCS}, the problem of optimizing privacy for a given $\sql$
constraint is studied, providing a
method to obtain a \sqlopt{\uprof}{\da}{\dq} mechanism for any $q,\pi,\dq,\da$.

In our case, we optimize $\sql$ for a given privacy constraint, constructing a
\dpopt{\prior}{\epsilon\dx}{\dq} mechanism.
We now show that, if $\dq$ and $\da$ coincide, the mechanism generated by any of the two optimization problems of
the previous sections is also \sqlopt{\uprof}{\dq}{\dq}.

$\adv$ corresponds to an adversary's remapping $H$ that minimizes his
expected error with respect to the metric $\da$ and
his prior knowledge $\prior$.
A crucial observation is that \priv{\dx} is closed under remapping.
%Since we would like to calculate the value of $\adv$ for our mechanism, we should first study how the attacker's remapping affects our mechanism.
%Since we would like to see the how much location privacy our mechanism offers with respect to this privacy definition, we should first study how the attacker's remapping affects our mechanism.
%It turns out that if a mechanism $K$ is \edpadj{\dx}, then $KH$ is also \edpadj{\dx}.
%It is also worth noting that if a mechanism $K$ is \edpadj{\dx}, then $KH$ is also \edpadj{\dx}.
%The following result shows that remapping does not violate .

\begin{restatable}{lemma}{theoremap}
\label{theo:remap}
Let  $K$ be a \privadj{\dx} mechanism, and let $H$ be a remapping. Then $KH$ is \privadj{\dx}.
\end{restatable}

%Now, suppose that a mechanism $K$ is \dpopt{\prior}{\dx}{\dq} for a given prior $\prior$ and metrics $\dx$ for the privacy definition and $\dq$ for the quality of service. We are interested in the value of $\advarg{K}{\prior}{\da}$, for a given adversary's metric $\da$.
%We can see that, if $\dq$ and $\da$ coincide, then the remapping $H$ for $K$ should not affect the value of $\adv$, and in fact, $\advarg{K}{\prior}{\dq}$ should be equal to $\sqlarg{K}{\prior}{\dq}$. Otherwise, the mechanism $KH$ would still be \edpadj{\dx} but it would have a lower $\sql$, contradicting the fact that $K$ is \dpopt{\prior}{\dx}{\dq}. This can be stated as follows:
Now let $K$ be a \dpopt{\uprof}{\dx}{\dq} mechanism and $H$ a
remapping. Since $KH$ is \privadj{\dx} (Lemma~\ref{theo:remap}) and $K$ is
optimal among all such mechanisms, we have that:
 \[
\sqlarg{K}{\uprof}{\dq} \leq \sqlarg{KH}{\uprof}{\dq}\quad \forall H
\]
As a consequence, assuming that $\dq$ and $\da$ coincide, the adversary
minimizes his expected error by applying no remapping at all (i.e. the identity
remapping), which means that $\advarg{K}{\pi}{\dq} = \sqlarg{K}{\pi}{\dq}$ and
therefore $K$ must be \sqlopt{\uprof}{\dq}{\dq}.
%
%\[
%\sqlarg{K}{\uprof}{\dq} = \advarg{K}{\uprof}{\dq}
%\]
%which in turn implies that the mechanism $K$ obtained from the optimization problem is \sqlopt{\uprof}{\dq}{\dq}.

\begin{restatable}{theorem}{theodpopt}
\label{theo:dpopt}
If a mechanism $K$ is \dpopt{\uprof}{\dx}{\dq} then it is also
\optpriv{\uprof}{\dq}{\dq}{q} for $q = \sqlarg{K}{\pi}{\dq}$.
\end{restatable}

It is important to note that Theorem \ref{theo:dpopt} holds for any metric
$\dx$. This means that both mechanisms obtained as result of the optimization
problems presented in Sections \ref{sec:approach} and
\ref{sec:mechanism-spanner} are \sqlopt{\uprof}{\dq}{\dq} -- since they are
\dpopt{\uprof}{\edx}{\dq} and \dpopt{\uprof}{\edg}{\dq} respectively -- however
for a different value of $q$. In fact, in contrast to the method of
\cite{Shokri:12:CCS} in which the quality bound $q$ is given as a parameter, our
method optimizes the $\sql$ given a privacy bound. Hence, the resulting
mechanism will be \sqlopt{\uprof}{\dq}{\dq}, but for a $q$ that is not known in
advance and will depend on the privacy constraint $\epsilon$ and the dilation
factor $\delta$. The greater the $\epsilon$ is (i.e. the higher the
privacy), or the lower the $\delta$ is (i.e. the better the approximation),
the lower the quality loss $q$ of the obtained mechanism will be.

Finally, we must remark that this result only holds in the case where the
metrics $\dq,\da$ coincide. If the metrics differ, e.g. the quality is measured
in terms of the Euclidean distance (the user is interested in accuracy) but the
adversary uses the binary distance (he is only interested in the exact
location), then this property will no longer be true.

%We are interested in the value of $\advarg{K}{\uprof}{\dx}$. We recall that, for any mechanism $K$ it holds that
%\[
%\privacy(K, \psi_u) \leq \sql(K, \psi_u)
%\]
%It turns out that the mechanism $\gioptd{\delta}{\psi_u}{\epsilon}$ achieves maximum privacy between all mechanisms with it's same $\sql$.
%
%\begin{theorem}
%Let $\psi_u$ be a user profile, $\delta$ a dilation factor and $\epsilon$ a privacy constraint. Then
%\[
%\privacy(\gioptd{\delta}{\psi_u}{\epsilon}, \psi_u) = \sql(\gioptd{\delta}{\psi_u}{\epsilon}, \psi_u)
%\]
%\end{theorem} 
%
%This means that $\giop$ achieves optimal location privacy with respect to the $\privacy$ metric.

\subsection{Practical considerations}

We conclude this section with a discussion on the practical applicability of
location obfuscation. First, it should be noted that, although constructing an
optimal mechanism is computationally demanding, once the matrix $K$ is computed,
obfuscating a location $x$ only involves drawing a reported location from the
distribution $K(x)$ which is computationally trivial. Moreover, although
obfuscation is meant to happen on the user's smartphone, computing the mechanism
can be offloaded to an external server and even parallelized. The user only
needs to transmit $\pi,\epsilon\dx,\dq$ (which are considered public) and
receive $K$, and the computation only needs to be performed occasionally, to adapt
to changes in the user profile.

Second, an important feature of obfuscation mechanisms is that they require no
cooperation from the service provider, who simply receives a location and has no
way of knowing whether it is real or not. Obfuscation can happen on the user's
device, at the operating system or browser level, which is crucial since the
user has strong incentives to apply it while the service provider does not.
The user's device could also perform filtering of the results, as described
in \cite{Andres:13:CCS}.

Finally, we argue that the common idea that users of LBSs are willing to give up their
privacy is misleading: the only alternative offered is not to use the service.
The usage of browser extensions such as ``Location Guard'' \cite{location.guard}
shows that users do care about their privacy and that obfuscation can be a
practical approach for using existing services in a privacy friendly way.

% !TEX root = geoopt.tex

\section{Evaluation}
\label{sec:evaluation}

 \begin{figure}[tb]
      \centering
      \includegraphics[width=1.0\columnwidth]{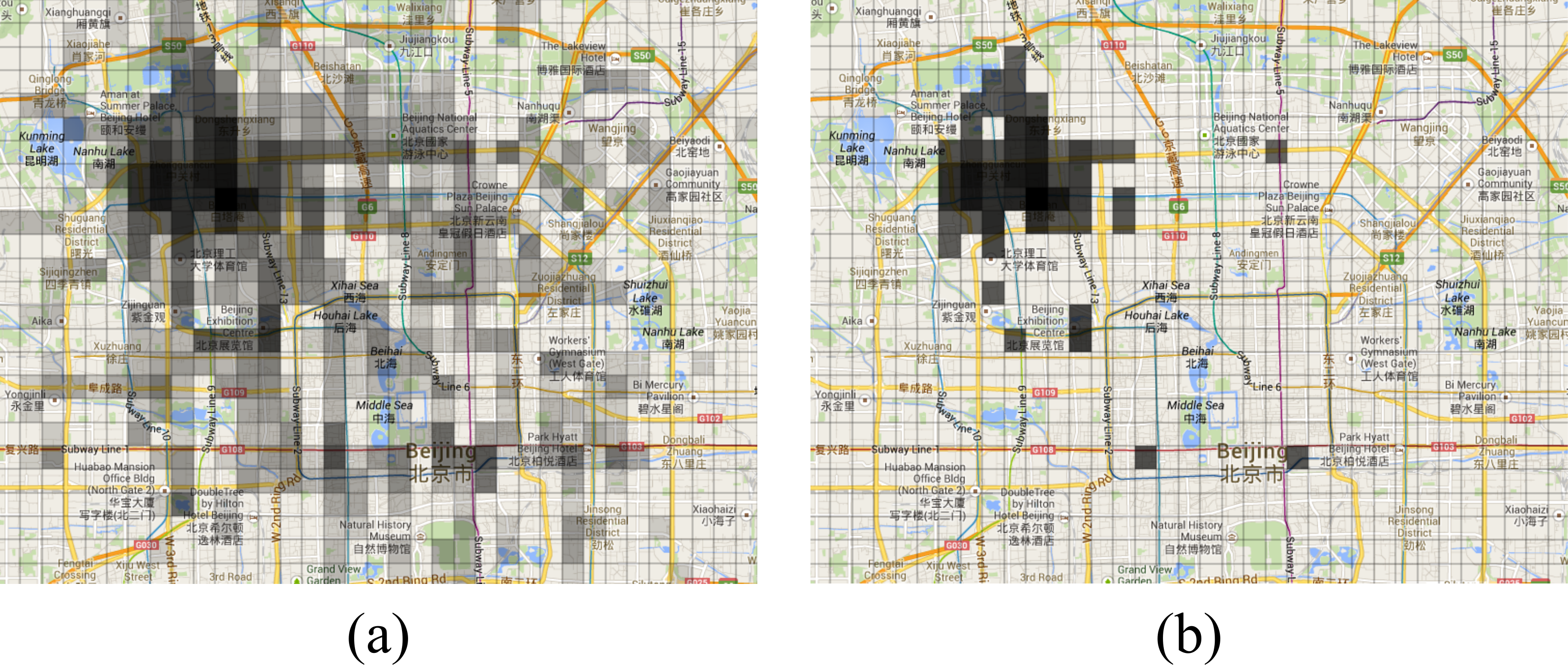}
   \caption{(a) Division of the map of Beijing into regions of size 0.658 x 0.712 km. The density of each region represents its ``score'', that is, how frequently users visit it. (b) The 50 selected regions. These regions are the ones with highest density between the whole set of regions.}
   \label{fig:regions}
% \vspace{-10pt}
 \end{figure}

%Given a set of locations $\calx$ in the map with a corresponding metric $\dx$, a dilation factor $\delta$, a user $u$  with user profile $\uprof$ and a privacy constraint $\epsilon$, we are able to get 
%a \dpopt{\uprof}{\edg}{\dq} mechanism for a given quality metric $\dq$.

In this section we evaluate the technique for constructing optimal mechanisms
described in the previous sections. We perform two kinds of evaluation: first,
a comparison with other mechanisms, namely the one of Shokri et al. and the
Planar Laplace mechanism. Second, a performance evaluation of the spanner
approximation technique.

The comparison with other mechanisms is performed with respect to both privacy
and quality loss. For privacy, the main motivation is to evaluate the
mechanisms' privacy under different priors, and in particular under priors
different than the one they were constructed with. Following the motivating
scenario of the introduction, we consider that a user's profile can vary
substantially between different time periods of the day, and simply by taking
into account the time of a query, the adversary can obtain a much more
informative prior which leads to a lower privacy. For the purposes of the
evaluation, we consider priors corresponding to four different time periods: the
full day, the morning (7am to noon), afternoon (noon to 7pm) and night (7pm to
7am). Then we construct the mechanisms using the full day prior and compare
their privacy for all time periods.

We perform our evaluation on two widely used datasets: GeoLife
\cite{Zheng:08:UbiComp,Zheng:09:WWW,DBLP:journals/debu/ZhengXM10} and T-Drive
\cite{Yuan:11:KDD, DBLP:conf/gis/YuanZZXXSH10}. The results of GeoLife are presented in detail in the following
sections, while, due to space restrictions, those of T-Drive (which are in
general similar) are summarized in Section \ref{sec:2nd-dataset}.

\subsection{The GeoLife dataset}
The GeoLife GPS Trajectories dataset contains 17621 traces from 182 users,
moving mainly in the north-west of Beijing, China, in a period of over five
years (from April 2007 to August 2012). The traces show users performing
routinary tasks (like going to and from work), and also traveling, shopping, and
doing other kinds of entertainment or unusual activities. Besides, the traces
were logged by users using different means of transportation, like walking,
public transport or bike. More than $90\%$ of the traces were logged in a dense
representation, meaning that the individual points in the trace were reported
every 1-5 seconds or every 5-10 meters. Since user behaviour changes over time,
and the mechanism should be occasionally reconstructed, we restrict each user's
traces to a 90 days period, and in particular to the one with the greatest number
of recorded traces, so that the prior is as informative as possible.

 \begin{figure}[t!]

%      \hspace{-0.5cm}
%      \includegraphics[width=\columnwidth]{figures/charts.pdf}
	 \hspace{-0.4cm}
      \includegraphics[width=1.1\columnwidth]{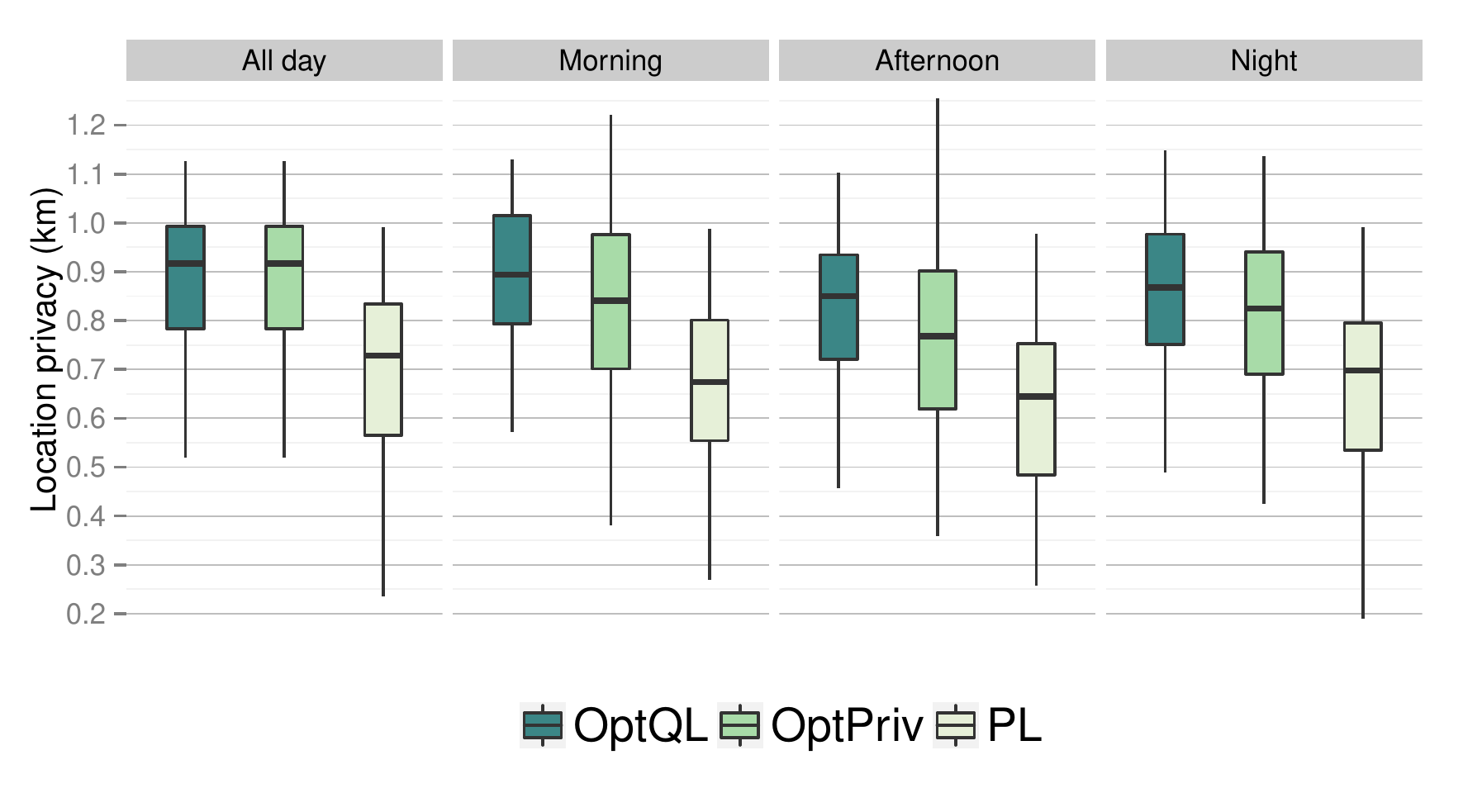}
         \vspace{-0.5cm}
   \caption{Boxplot of the location privacy provided by the three different mechanisms under considered priors. The $\giop$ mechanism was constructed with $\epsilon = 1.07$ and $\delta = 1.05$. }
   \label{fig:chart-privacy}
 \end{figure}

\subsection{Mechanism comparison wrt privacy and quality loss}
\label{sec:eval-privacy}

For the evaluation, we divide the map of Beijing into a grid of regions 0.658
km wide and 0.712 km high, displayed in Figure \ref{fig:regions}a. To avoid
users for which little information is available, we only keep those having
at least 20 recorded points within the grid area for each one of the time
periods. Whenever we count points, those falling within the same grid region during the
same hour are counted only once, to prevent traces with a huge number of
points in the same region (e.g. the user's home) from completely skewing the
results. After this filtering, we end up with 116 users (64\% of the total 182).

We then proceed to calculate the 50 ``most popular'' regions of the grid as
follows: for each user, we select the 30 regions in which he spends the
greatest amount of time. A region's ``score'' is the number of users that have
it in their 30 highest ranked ones. Then we select the 50 regions with the
highest score.

Figure \ref{fig:regions}a shows the division of the map into regions, with the
opacity representing the score of each of them, while Figure \ref{fig:regions}b
shows the 50 regions with highest score. We can see that most of the selected
regions are located in the south-east of the Haidian district, and all of them
are located in the north-west of Beijing. We consider the set of locations
$\calx$ to be the centers of the selected regions, and the metric $\dx$ to be
the Euclidean distance between these centers, i.e. $\dx = d_2$.

Finally, a second filtering is performed, again keeping users with at least
20 points in each time period, but this time considering only the 50 selected
regions. After this, we end up with a final set of 86 users (46\% of the total 182).

In this section, we evaluate the location privacy and the utility of three different mechanisms under the several prior distributions for each user. These priors correspond to different parts of the day (all day, morning, afternoon and night), and are computed by counting the number of points, logged in the corresponding time period, that fall in each of the selected regions (again, counting only once those points logged within the same hour), and then by normalizing these numbers to obtain a probability distribution.

 \begin{figure}[t]
      \centering
      \includegraphics[width=\columnwidth]{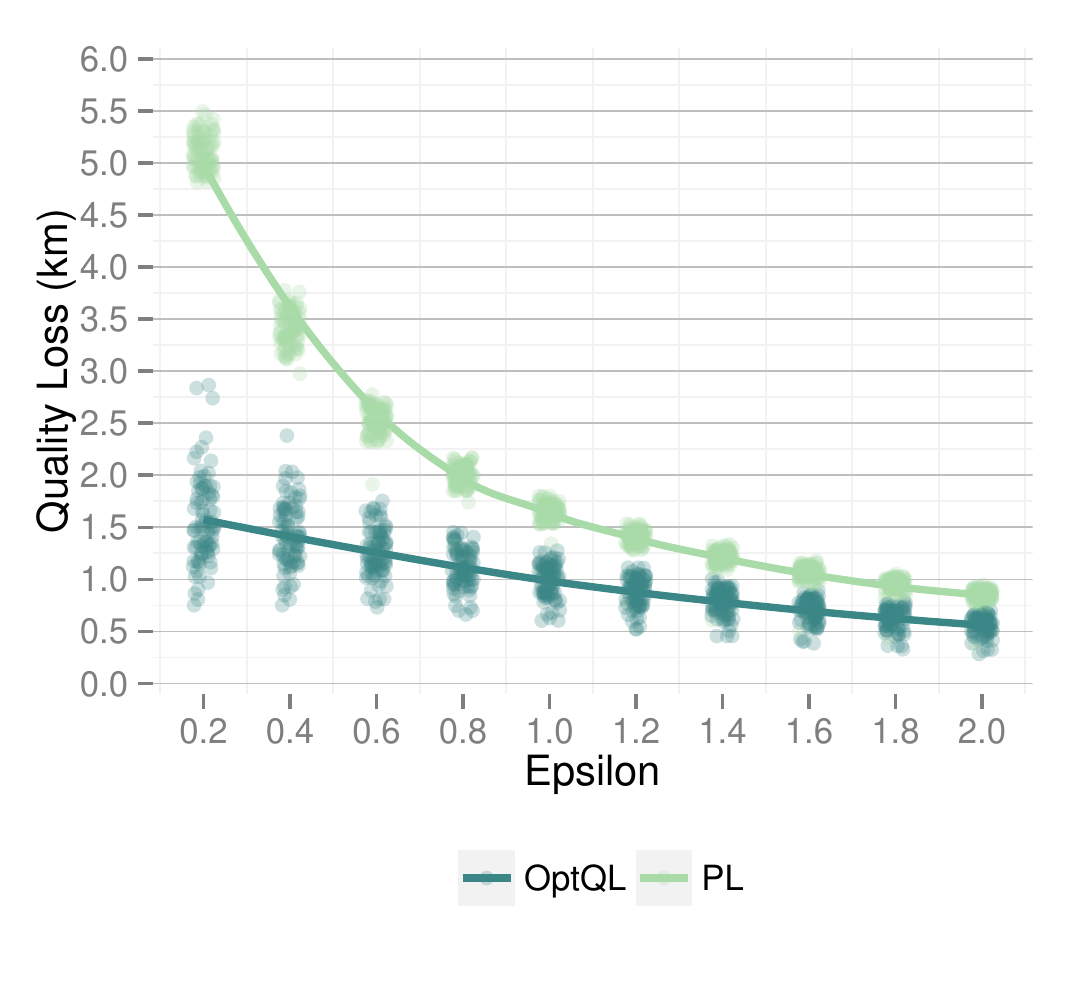}
   \vspace{-0.75cm}
   \caption{Quality loss of the $\giop$ and $\plap$ mechanisms for different values of $\epsilon$. The mechanisms were calculated for all users. Here, points represent the utility for every user, while the two lines join the medians for each mechanism and each value of $\epsilon$.}
   \label{fig:chart-epsilons}

 \end{figure}

We start by evaluating the location privacy provided by the different mechanisms.
However, we must note that in general location privacy mechanisms do not satisfy
\edp{\dx} unless they are specifically designed to do so. Therefore, for this evaluation, we measure
location privacy with the metric $\adv$, proposed in \cite{Shokri:12:CCS} and
described in Section \ref{sec:rezas-mech}, which measures the expected error of
the attacker under a given prior distribution. 
In order to perform a fair comparison, we construct the mechanisms in such a way that their $\sql$ coincide.
The first step is to select a privacy level $\epsilon$ and a dilation $\delta$, and then to construct the mechanism 
described in Section \ref{sec:mechanism-spanner}. We will call this mechanism $\giop$.
This mechanism has a $\sql$ of
$q = \sqlarg{\giop}{\uprof}{d_2}$.
We then continue by constructing the optimal mechanism of Shokri et al \cite{Shokri:12:CCS}, and setting the $\sql$ as $q$. We call this mechanism $\eeop$.
Finally, we compute a discretized version of the Planar Laplace mechanism
of Andr\'es et al \cite{Andres:13:CCS}.
under a privacy constraint $\epsilon'$ (in general different from $\epsilon$) such that the $\sql$ of this mechanism is also $q$. We call this mechanism $\plap$. Note that at the end of this process, by construction, the $\sql$ of the three mechanisms is $q$.

 \begin{figure*}[t]
%      \centering
%      \includegraphics[width=1.0\columnwidth]{figures/utility-and-constraintsv2.pdf}
	\includegraphics[width=\textwidth]{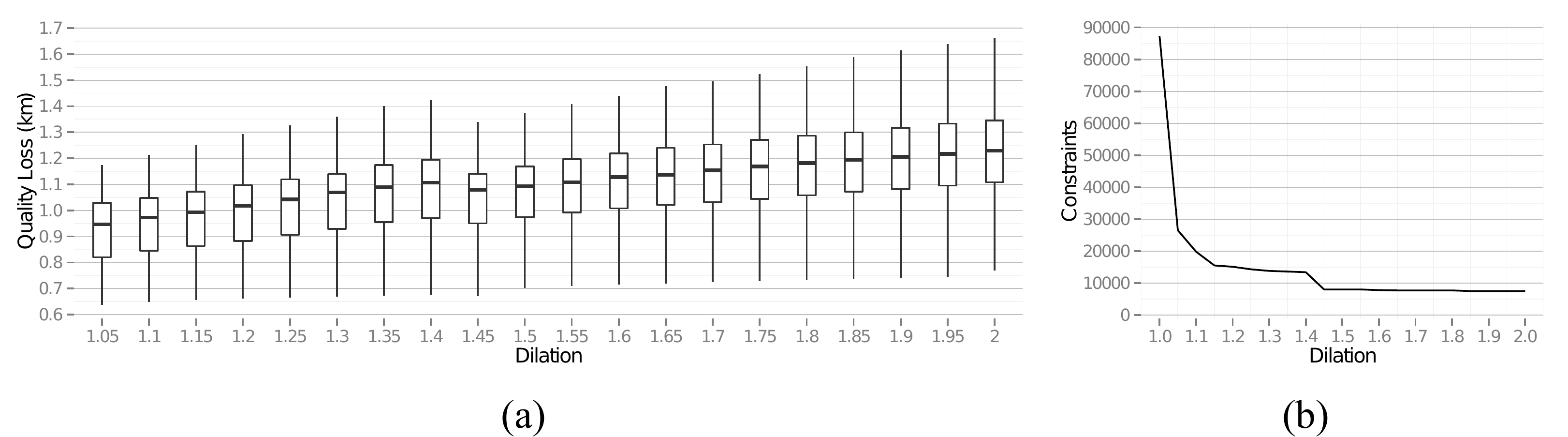}
   \caption{(a) Boxplot of the relation between $\sql$ and dilation for the mechanism $\giop$ with privacy constraint $\epsilon = 1.07$. The spanner is calculated with the greedy algorithm presented in Section \ref{sec:algorithm}. (b) Relation between the approximation ratio and the number of constraints in the linear program. This number is independent from the user and form the value of $\epsilon$.}
   \label{fig:charts-utilities}
 \end{figure*}

We begin the evaluation comparing the location privacy of each mechanism for each of the selected users, under the four constructed priors. We fix $\epsilon = 1.07$ (which intuitively corresponds to a ratio of 2 between the probability for two regions adjacent in the grid to report the same observed location) and $\delta = 1.05$. Figure \ref{fig:chart-privacy} shows a boxplot of the location privacy (in km) offered by the different mechanisms under each prior. In all four cases, the general performance of our mechanism is better than that of the others, with the only exception
being the all-day prior (which is the one used in the construction of the mechanisms) since, as explained in Section \ref{sec:mech-ee}, $\giop$ and $\eeop$ are \sqlopt{\uprof}{d_2}{d_2}
and therefore offer the same privacy.

Finally, to show the benefits of using a mechanism with optimal utility, we compare now the $\sql$ of the mechanisms $\giop$ and $\plap$ when both mechanisms are generated with the same privacy level $\epsilon$. We can see the results in Figure \ref{fig:chart-epsilons}. The $\giop$ mechanism clearly offers a better utility to the user, while guaranteeing the same level of geo-indistinguishability.

\subsection{Performance of the approximation algorithm}
\label{sec:performance}

We recall from Section \ref{sec:mechanism-spanner} that if we consider a large number of locations in $\calx$, then the number of constraints in the linear program might be large. Hence, we introduced a method 
based on a spanning graph $G$ to reduce the total number of constraints of the linear program. 
However, in general the obtained mechanism is no longer \dpopt{\uprof}{\edx}{\dq}, and therefore it has a higher $\sql$ than the optimal one.

In this section we study the tradeoff between the increase in the $\sql$ of the mechanism and the reduction in the number of constraints of the optimization problem, as a consequence of using our approximation technique. We also show how this reduction affects the running time of the whole approach.
We start by constructing the $\giop$ mechanism for all selected users and for different dilations in the range from $1.05$ to $2.0$, in all cases considering $\epsilon = 1.07$ as before.
We then measure the $\sql$ of each mechanism under the user profile. We can see the results in Figure \ref{fig:charts-utilities}a. 
It is clear that the $\sql$ increases slowly with respect to the dilation: the median value is $0.946$ km for $\delta = 1.05$, is $0.972$ km for $\delta = 1.1$, and $1.018$ km for $\delta = 1.2$. Therefore we can deduce that, for a reasonable approximation, the increase in the quality loss is not really significant.
%For instance, the mechanism for the first user of the dataset has:
%\begin{itemize}
%\item $\sql$ of $0.98$ km for $\delta = 1$ (the optimal case).
%\item $\sql$ of $1$ km for $\delta = 1.05$ ($2\%$ increase).
%\item $\sql$ of $1.03$ km for $\delta = 1.1$ ($5\%$ increase).
% \end{itemize}
It is worth noting that we do not show the $\sql$ for $\delta = 1$ in the plot (corresponding to the case where $\dx$ and $d_G$ are the same). The reason is that in that case the number of constraints is really high, and therefore it takes a lot of time to generate one instance of the mechanism (and much more time to generate it for the 86 users considered). 

The relation between the dilation and the number of constraints is shown in Figure \ref{fig:charts-utilities}b. Note that this number is independent from the user, and therefore it is enough to calculate it for just one of them. It is clear that the number of constraints decreases exponentially with respect to the dilation, and therefore even for small dilations (which in turn mean good approximations) the number of constraints is significantly reduced with the proposed approximation technique. For instance, we have 87250 constraints for $\delta = 1$ (the optimal case), and 25551 constraints for $\delta = 1.05$. This represents a decrease of $71\%$ with respect to the optimal case, with only $1.05$ approximation ratio.
%Particularly, we have:
%\begin{itemize}
%\item 87250 constraints for $\delta = 1$ (the optimal case).
%\item 25551 constraints for $\delta = 1.05$ ($71\%$ decrease).
%\item 19151 constraints for $\delta = 1.1$ ($78\%$ decrease).
% \end{itemize}

It is also worth noting that, between $\delta = 1.4$ and $\delta = 1.45$ there is a pronounced decrease in the number of constraints (Figure \ref{fig:charts-utilities}b) and \emph{also} a decrease in the $\sql$ (Figure \ref{fig:charts-utilities}a). This might seem counterintuitive at first, since one would expect that a worse approximation should always imply a higher loss of quality. However, there is a simple explanation: although the spanner with $\delta = 1.45$ has a higher worst-case approximation ratio, the average-case ratio is actually better that the one of the spanner with $\delta = 1.4$. This phenomenon is a consequence of the particular topology of the set of locations and to the algorithm used to get the spanner.

Finally, we measure the running time of the method used to generate the $\giop$ mechanism, under different methods to solve the linear optimization problem. 
The experiments were performed in a 2.8 GHz Intel Core i7 MacBook Pro with 8 GB of RAM running Mac OS X 10.9.1, and the source code for the method was written in C++, using the routines in the GLPK library for the linear program. We compare the performance of three different methods included in the library: the simplex method in both its primal and dual form, and the primal-dual interior-point method. Besides, we run these methods on both the primal linear program presented in Section \ref{sec:mechanism-spanner} and its dual form, presented in Appendix \ref{app:dual}.
Since the running time depends mainly on the number of locations being considered, in the experiments we focus on just one user of the dataset, and we fix the privacy level as $\epsilon = 1.07$. The results can be seen in Table \ref{tab:results}.
%The fields marked with X correspond to cases where the execution took more than one hour, after which it was stopped.
Some fields are marked with ``1h+'', meaning that the execution took more than one hour, after which it was stopped. Others are marked with ``Error'', meaning that the execution stopped before one hour with an error\footnote{The actual error message in this case was: ``Error: unable to factorize the basis matrix (1). 
Sorry, basis recovery procedure not implemented yet''}.
A particular case of error happened when running the interior-point method on the dual linear program, where all executions ended with a ``numerical instability'' error (and therefore this case is not included in the table). 
%We can see that the only two methods that behave consistently (that never finish with error, and the running time increases when the dilation decreases) are the dual simplex and the interior-point methods, both when applied to the primal program. 
%From these,  we can observe that the interior-point method performs better in the case of bigger dilation, while it does it much worse for very small ones. This tendency can be appreciated more clearly when considering a bigger number of locations, in this case 75.
From the results we can observe that:
\begin{itemize}
\item The only two methods that behave consistently (that never finish with error, and the running time increases when the dilation decreases) are the dual simplex and the interior-point methods, both when applied to the primal program.
\item From these, the interior-point method performs better in the case of bigger dilation, while it does it much worse for very small ones.
\item Somewhat surprisingly, the dual linear program does not offer a significant performance improvement, specially when compared with the interior-point method.
\end{itemize}

In the case of $\eeop$, the mechanism is generated using Matlab's linear program solver (source code kindly provided by the authors of \cite{Shokri:12:CCS}). We generated the mechanism for the same cases, and observed that the running time mainly depends on the number of regions: for 50 regions, the mechanism is generated in approximately 1 minute, while for 75 regions it takes about 11 minutes.

\begin{table}[t]
\footnotesize
	\small
	\begin{tabular} { cc|c|c|c|c|c|c| }
	
	\cline{3-7}
	& & \multicolumn{2}{c|}{Primal simplex} & \multicolumn{2}{c|}{Dual simplex} & Interior \\ \hline
\multicolumn{1}{|c|}{$|\calx|$} & $\delta$ & Pr. LP & Du. LP & Pr. LP & Du. LP & Pr. LP \\ 
	\hline
	\multicolumn{1}{|c|}{ \multirow{5}{*}{\begin{tabular}[x]{@{}c@{}}50\end{tabular}} }& \multicolumn{1}{|c|}{ $1.0$ }
												  & 57s	& 1h+ 		& 40s 	& 45s 	& 49m 20s \\ 
	\multicolumn{1}{|c|}{ }& \multicolumn{1}{|c|}{ $1.1$ } & 46.4s 		& 5.2 	& 5.9s 	& 15.5s 	& 7.5s \\
	\multicolumn{1}{|c|}{ }& \multicolumn{1}{|c|}{ $1.2$ } & 4m 37s 	& 2s		& 4s	 	& 1h+ 		& 2.7s \\ 
	\multicolumn{1}{|c|}{ }& \multicolumn{1}{|c|}{ $1.5$ } & 2s 		& 1s 		& 2s 		& 3s 		& 0.5s \\ 
	\multicolumn{1}{|c|}{ }& \multicolumn{1}{|c|}{ $2.0$ } & Error 		& 1s 		& 2s 		& 2s 		& 0.5s \\ \hline
	\multicolumn{1}{|c|}{ \multirow{5}{*}{\begin{tabular}[x]{@{}c@{}}75\end{tabular}} }& \multicolumn{1}{|c|}{ $1.0$ }
												   & 1h+	& 1h+ 		& 29m 26s& 1h+ 	& 1h+	 \\ 
	\multicolumn{1}{|c|}{ }& \multicolumn{1}{|c|}{ $1.1$ } & 1h+		& Error	& 1m 12s 		& 2m 19s 	& 55s  \\ 
	\multicolumn{1}{|c|}{ }& \multicolumn{1}{|c|}{ $1.2$ } & 1h+ 		& Error	& 42s	 	& 48.4s 	& 11.7s \\ 
	\multicolumn{1}{|c|}{ }& \multicolumn{1}{|c|}{ $1.5$ } & 1h+ 		& 5m 55s 	& 19.2s 		& 1h+ 		& 2.2s \\
	\multicolumn{1}{|c|}{ }& \multicolumn{1}{|c|}{ $ 2.0$ } & 1h+ 		& 21.8s 	& 27.2s 		& 15.5s 	& 1.7s  \\ \hline
	\end{tabular}
	\caption{Execution times of our approach for 50 and 75 locations, for different values of $\delta$, and using different methods to solve the linear program.}
	\label{tab:results}
\end{table}

%2nd dataset:
%first filter: 417 users
%second filter: 203 users

 \begin{figure}[t]

%      \hspace{-0.5cm}
%      \includegraphics[width=\columnwidth]{figures/charts.pdf}
%	  \hspace{-0.4cm}
\centering
      \includegraphics[width=\columnwidth]{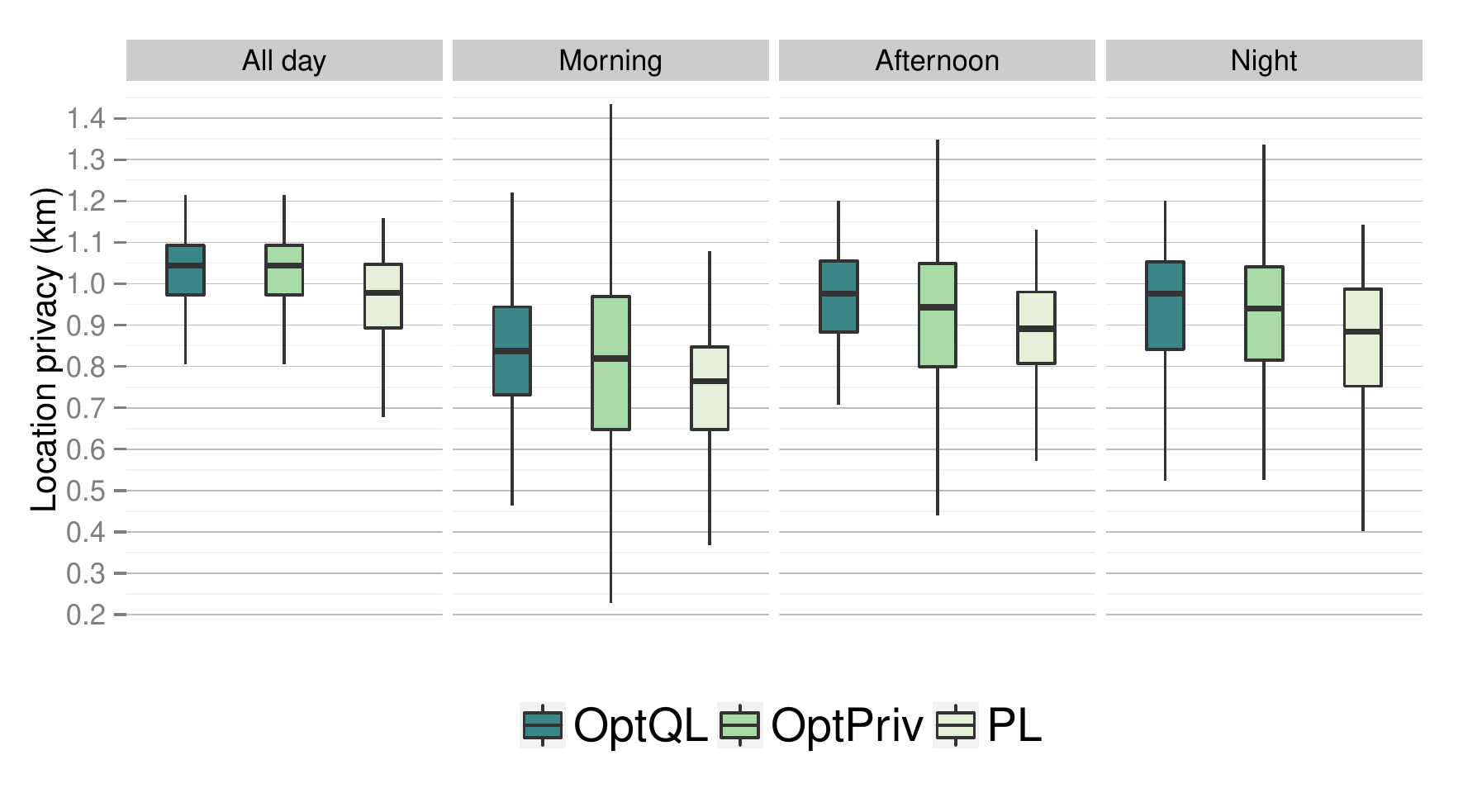}
%         \vspace{-0.5cm}
   \caption{Boxplot of the location privacy for the T-Drive dataset. The median value of the location privacy for $\giop$ is always as good as the one of the other mechanisms. }
   \label{fig:chart-privacy-2nd}
 \end{figure}

 \begin{figure}[b!]

%      \hspace{-0.5cm}
%      \includegraphics[width=\columnwidth]{figures/charts.pdf}
      \centering
      \includegraphics[width=\columnwidth]{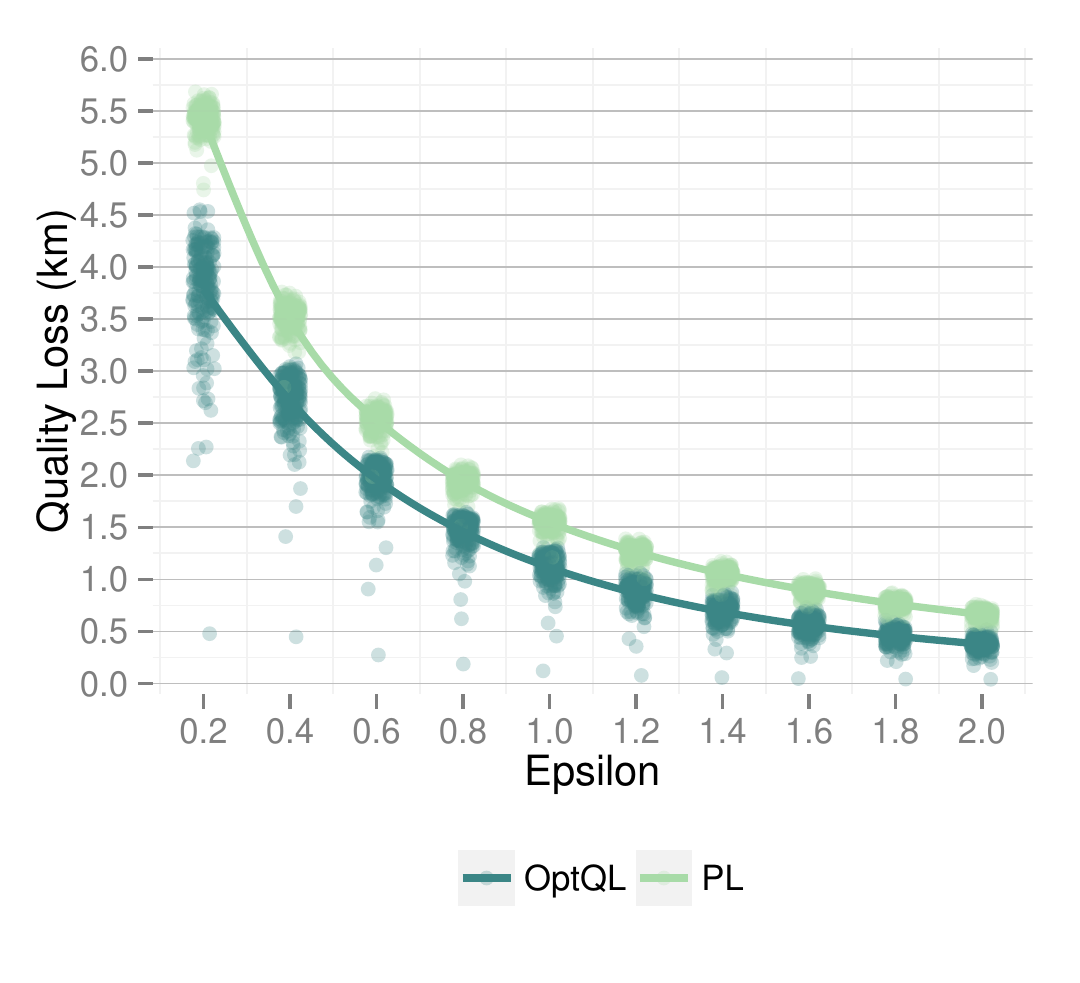}
         \vspace{-0.5cm}
   \caption{Quality loss of the $\giop$ and $\plap$ mechanisms for different values of $\epsilon$, using the data in the T-Drive dataset. The loss of quality of $\giop$ is always smaller than the one of $\plap$. }
   \label{fig:chart-epsilons-2nd}
 \end{figure}

\subsection{The T-Drive dataset}
\label{sec:2nd-dataset}

In order to reaffirm the validity of the proposed approach, we performed the same evaluation in a different dataset: the T-Drive trajectories dataset. This dataset contains traces of 10357 taxis in Beijing, China, during the period of one week. The total distance of the traces in this dataset is about 9 million kilometres, with more than 15 million reported points. The average time between consecutive points in a trace is 177 seconds, and the average distance is 623 meters.

Due to the huge amount of users in this dataset, we started the evaluation process by blindly selecting (using a standard random function) 5\% of the total number users (about 532 users out of 10357). We then perform the same steps as described in the previous sections, particularly those described in Section \ref{sec:eval-privacy}. In Figure \ref{fig:chart-privacy-2nd} we can see the comparison of the location privacy for the different mechanisms. We can see that, also for this dataset, the privacy level of $\giop$ is, in general, as good as the one of $\eeop$, and always better than the one of $\plap$. In particular, the median value for $\giop$ is always higher than the corresponding one for the other mechanisms (again, with the exception of the all day prior, for which we know that these values coincide).  We can also see in Figure \ref{fig:chart-epsilons-2nd} the comparison in terms of utility of the mechanisms $\giop$ and $\plap$. Again, the quality loss of $\giop$ is, in all cases, better than the one of $\plap$. This is to be expected, since, from all mechanisms providing a certain geo-indistinguishability, $\giop$ is the one with optimal utility (or really close to the optimal utility when the approximation is used).

%
% \begin{figure}[t]
%
%%      \hspace{-0.5cm}
%%      \includegraphics[width=\columnwidth]{figures/charts.pdf}
%%	  \hspace{-0.4cm}
%\centering
%      \includegraphics[width=\columnwidth]{figures/plot_epsilons_2nd_dataset.pdf}
%         \vspace{-0.75cm}
%   \caption{Relation between the privacy level $\epsilon$ and the quality loss of the $\giop$ and $\plap$ mechanisms for the T-Drive dataset.  }
%   \vspace{-0.1cm}
%   \label{fig:chart-epsilons-2nd}
% \end{figure}
%
%%\break

% !TEX root = geoopt.tex

\section{Conclusion and related work}\label{related}

%%%
%- start by saying sth like "there was a lot  of work in recent years about location privacy"
%- mention works that aim at hiding the identity, and tell that we are not interested in that
%- talk about mechanisms with no clearly defined privacy guarantee
%- talk about mechanisms that have a clear privacy guarantee
%- Talk about DP, and the works in which it has been applied to loc. privacy

\subsubsection*{Related work}
In the last years, a large number of location-privacy protection techniques, diverse  both in nature and goals, have been   proposed and studied.
Many of these aim at allowing the user of an LBS to hide his \emph{identity} from the service provider. Several approaches are based in the notion of $k$-anonymity \cite{Gruteser:03:MobiSys,Gedik:05:ICDCS,Mokbel:06:VLDB}, requiring that the attacker cannot identify a user from at least other $k-1$ different users. Others are based on the idea of letting the users use pseudonyms to interact with the system, and on having regions (\emph{mix zones}, \cite{Beresford:03:Perv,Freudiger:09:PETS}),  where the users can change their pseudonyms without being traced by the system. All these approaches are incomparable with ours, since ours aims at hiding the \emph{location} of the user and not his identity.

Many approaches to location privacy are based on  obfuscating the position of
the user. A common technique for this purpose is \emph{cloaking}
\cite{Bamba:08:WWW,Duckham:05:Pervasive,Xue:09:LoCa,Gedik:05:ICDCS}, which
consists in blurring the user's location by reporting a region to the service
provider. 
%This region must in general satisfy some property in order to preserve the
%location privacy of the user. 
Another technique is based on adding \emph{dummy
locations}\cite{Kido:05:ICDE,Shankar:09:UbiComp,Chow:09:WPES} to the request
sent to the service provider. In order to preserve privacy, these dummy
locations should be generated in such a way that they look equally likely to be
the user's real position. A different approach is to construct mechanisms that
provide optimal privacy under certain quality constraints \cite{Shokri:12:CCS}
(an approach dual to ours, as discussed in the introduction),
while \cite{Herrmann:13:WPES} additionally takes into account bandwidth constraints.
Finally, collaborative models have been proposed
\cite{Shokri:14:TDSC}, where privacy is achieved with a peer-to-peer scheme
where users avoid querying the service provider whenever they can find the
requested information among their peers.

Differential Privacy has also been used in the context of location privacy;
however, it is in general used to protect \emph{aggregate} location information.
For instance, \cite{Machanavajjhala:08:ICDE} presents a way to statistically
simulate the location data from a database while providing privacy guarantees.
In \cite{Ho:11:GIS}, a quad tree spatial decomposition technique is used to
achieve differential privacy in a database with location patter mining
capabilities. On the other hand, Dewri \cite{Dewri:12:TMC} proposes a
combination of differential privacy and $k$-anonymity for the purposes of hiding
the location of a single individual. The proposed definition requires that the
distances between the probability distributions corresponding to $k$ fixed
locations (defined as the anonymity set) should not be greater than the privacy
parameter $\epsilon$.

The work  closest to ours is \cite{Shokri:14:TechRep}, which independently
proposes a linear programming technique to construct optimal obfuscation
mechanisms wrt either $\adv$ or \geoind{}. Although there is an overlap in the
main construction (the optimization problem of Section~\ref{sec:approach}), most
of the results are substantially different. The approximation technique of
\cite{Shokri:14:TechRep} consists of discarding some of the \geoind{}
constraints when the distance involved is larger than a certain lower bound.
This affects the \geoind{} guarantees of the mechanism, although the effect can
be tuned by properly selecting the bound for discarding constraints. On the
other hand, our approximation technique, based on spanning graphs, can be used
to reduce the number of constraints from cubic to quadratic without jeopardizing
the privacy guarantees, by accepting a small decrease on the utility. Moreover,
we show that the mechanism obtained from this optimization problem is also
optimal wrt $\adv$ (Theorem~\ref{theo:dpopt}), which is an important property of
the proposed method. Finally, the evaluation methods are substantially different: in
\cite{Shokri:14:TechRep} the employed set of prior distributions differ in their
level of entropy (priors with low entropy are considered more informative). In
our work, we obtain the different priors by combining the distribution of the
user (assumed to be known by the adversary) with some public available
information (for instance, the time of the day).
%This is a natural way of enhancing the prior, and it must be noted that the
%obtained distribution is, in general, more accurate, but at the same time it
%might have less entropy than the original one.

Finally, \priv{\dx} has been used in \cite{Dwork:12:ITCS} to capture
\emph{fairness}, instead of privacy. The goal is to construct a fair
mechanism that produces similar reported values for ``similar'' users, the
similarity being captured by the metric. As in our work, the construction
involves solving an optimization problem, however no technique is used to
reduce the number of constraints. 

%Another work closely related to ours is \cite{Dwork:12:ITCS} in which an
%extended definition of differential privacy is used to capture the notion of
%fairness in classification. A metric $d$ is used to model the fact that certain
%individuals are required to be classified similarly, and a mechanism satisfying
%\priv{d} is considered fair, since it produces similar results for similar
%individuals. We view fairness as one of the many interesting notions that can be
%obtained through the use of metrics in various contexts, thus it encourages our
%goal of studying \priv{d}. With respect to the actual metrics used in this
%paper, the difference is that we consider metrics that depend on the
%individuals' values, while \cite{Dwork:12:ITCS} considers metrics between
%individuals.

% !TEX root = geoopt.tex

\subsubsection*{Conclusion}

%- combine the advantages of the approaches geo-ind and Reza's
%- technique to reduce the running time of the method
%- optimallity in both senses (with disclaimer)
%- case study based on real data that shows that uor approach not only
%performs better privacy-wise (specially when the prior of the attacker
%differs from the user profile) but also in running time (for reasonably low dilations).

In this paper we have  developed a method to generate a mechanism for location privacy that combines the advantages of the geo-indistinguishability privacy guarantee of \cite{Andres:13:CCS} and the optimal mechanism of  \cite{Shokri:12:CCS}. 
%In our approach, a fixed privacy level in terms of geo-indistinguishability is defined in advance, and then a mechanism with optimal utility (one that minimizes the service quality loss for the user) is generated by solving a linear optimization problem. Besides, the privacy guarantee of the obtained mechanism is not affected by the side knowledge of the attacker. 
Since linear optimization is computationally demanding, we have provided  a technique to reduce the total number of constraints in the linear program, based on the use of a spanning graph to approximate distances between locations, which allows a huge reduction on the number of constraints with only a small decrease in the utility. 
%Moreover, in the case where the metric used by the adversary and the one used to calculate the utility coincide, then the obtained mechanism is also optimal in the sense described in \cite{Shokri:12:CCS}, i.e. when the privacy metric is the expected error of the attacker. 
Finally, we have evaluated the proposed approach using traces from real users, and we have compared both the privacy and the running time of our mechanism with that of  \cite{Shokri:12:CCS}. It turns out that our mechanism offers better privacy guarantees when the side knowledge of the attacker is different from the distribution used to construct the mechanisms. Besides, for a reasonably good approximation factor, we have showed that our approach performs much better in terms of running time.

\section{Acknowledgements}
This work was partially supported by the MSR-INRIA joint lab, by the European Union 
7th FP project MEALS, by the project ANR-12-IS02-001 PACE, and by the 
INRIA Large Scale Initiative CAPPRIS. The work of Nicol\'as E. Bordenabe was partially funded by the DGA.
%French Defense 
%Procurement Agency (DGA) by a PhD grant.

%\break

%\newpage
\bibliographystyle{splncs}
\bibliography{short}

\begin{thebibliography}{10}

\bibitem{Freudiger:11:FC}
Freudiger, J., Shokri, R., Hubaux, J.P.:
\newblock Evaluating the privacy risk of location-based services.
\newblock In: Proc. of FC'11. Volume 7035 of LNCS., Springer (2011)  31--46

\bibitem{Golle:09:PC}
Golle, P., Partridge, K.:
\newblock On the anonymity of home/work location pairs.
\newblock In: Proc. of PerCom'09. Volume 5538 of LNCS.
\newblock Springer-Verlag (2009)  390--397

\bibitem{Krumm:07:PC}
Krumm, J.:
\newblock Inference attacks on location tracks.
\newblock In: Proc. of PERVASIVE. Volume 4480 of LNCS., Springer (2007)
  127--143

\bibitem{Beresford:03:Perv}
Beresford, A.R., Stajano, F.:
\newblock Location privacy in pervasive computing.
\newblock IEEE Pervasive Computing \textbf{2}(1) (2003)  46--55

\bibitem{Chow:09:WPES}
Chow, R., Golle, P.:
\newblock Faking contextual data for fun, profit, and privacy.
\newblock In: Proc. of WPES, ACM (2009)  105--108

\bibitem{Freudiger:09:PETS}
Freudiger, J., Shokri, R., Hubaux, J.P.:
\newblock On the optimal placement of mix zones.
\newblock In: Proc. of PETS 2009. Volume 5672 of LNCS., Springer (2009)
  216--234

\bibitem{Hoh:07:CCS}
Hoh, B., Gruteser, M., Xiong, H., Alrabady, A.:
\newblock Preserving privacy in gps traces via uncertainty-aware path cloaking.
\newblock In: Proc. of CCS, ACM (2007)  161--171

\bibitem{Shokri:12:CCS}
Shokri, R., Theodorakopoulos, G., Troncoso, C., Hubaux, J.P., Boudec, J.Y.L.:
\newblock Protecting location privacy: optimal strategy against localization
  attacks.
\newblock In: Proc. of CCS, ACM (2012)  617--627

\bibitem{Andres:13:CCS}
Andr{\'e}s, M.E., Bordenabe, N.E., Chatzikokolakis, K., Palamidessi, C.:
\newblock Geo-indistinguishability: differential privacy for location-based
  systems.
\newblock In: Proc. of CCS, ACM (2013)  901--914

\bibitem{Shokri:11:SP}
Shokri, R., Theodorakopoulos, G., Boudec, J.Y.L., Hubaux, J.P.:
\newblock Quantifying location privacy.
\newblock In: Proc. of S\&P, IEEE (2011)  247--262

\bibitem{Dwork:06:TCC}
Dwork, C., Mcsherry, F., Nissim, K., Smith, A.:
\newblock Calibrating noise to sensitivity in private data analysis.
\newblock In: Proc. of TCC. Volume 3876 of LNCS., Springer (2006)  265--284

\bibitem{Chatzikokolakis:13:PETS}
Chatzikokolakis, K., Andr{\'e}s, M.E., Bordenabe, N.E., Palamidessi, C.:
\newblock {Broadening the scope of Differential Privacy using metrics}.
\newblock In: Proc. of PETS. Volume 7981 of LNCS., Springer (2013)  82--102

\bibitem{Chatzikokolakis:14:PETS}
Chatzikokolakis, K., Palamidessi, C., Stronati, M.:
\newblock A predictive differentially-private mechanism for mobility traces.
\newblock In: Proc. of PETS. Volume 8555 of LNCS., Springer (2014)  21--41

\bibitem{Reed:10:ICFP}
Reed, J., Pierce, B.C.:
\newblock Distance makes the types grow stronger: a calculus for differential
  privacy.
\newblock In: Proc. of ICFP, ACM (2010)  157--168

\bibitem{Narasimhan:07:BOOK}
Narasimhan, G., Smid, M.:
\newblock Geometric spanner networks.
\newblock CUP (2007)

\bibitem{Sack:99:BOOK}
Sack, J., Urrutia, J.:
\newblock Handbook of Computational Geometry.
\newblock Elsevier Science (1999)

\bibitem{Klein:06:GD}
Klein, R., Kutz, M.:
\newblock {Computing Geometric Minimum-Dilation Graphs is NP-Hard}.
\newblock In: Proc. of the GD. Volume 4372., Springer (2006)  196--207

\bibitem{location.guard}
:
\newblock Location Guard. \url{https://github.com/chatziko/location-guard}.

\bibitem{Zheng:08:UbiComp}
Zheng, Y., Li, Q., Chen, Y., Xie, X., Ma, W.Y.:
\newblock {Understanding Mobility Based on GPS Data}.
\newblock In: Proc. of UbiComp 2008. (2008)

\bibitem{Zheng:09:WWW}
Zheng, Y., Zhang, L., Xie, X., Ma, W.Y.:
\newblock {Mining interesting locations and travel sequences from GPS
  trajectories}.
\newblock In: Proc. of WWW 2009. (2009)

\bibitem{DBLP:journals/debu/ZhengXM10}
Zheng, Y., Xie, X., Ma, W.Y.:
\newblock Geolife: A collaborative social networking service among user,
  location and trajectory.
\newblock IEEE Data Eng. Bull. \textbf{33}(2) (2010)  32--39

\bibitem{Yuan:11:KDD}
Yuan, J., Zheng, Y., Xie, X., Sun, G.:
\newblock Driving with knowledge from the physical world.
\newblock In: {The 17th ACM SIGKDD international conference on Knowledge
  Discovery and Data mining, KDD '11}. (2011)

\bibitem{DBLP:conf/gis/YuanZZXXSH10}
Yuan, J., Zheng, Y., Zhang, C., Xie, W., Xie, X., Sun, G., Huang, Y.:
\newblock T-drive: driving directions based on taxi trajectories.
\newblock In: GIS. (2010)  99--108

\bibitem{Gruteser:03:MobiSys}
Gruteser, M., Grunwald, D.:
\newblock Anonymous usage of location-based services through spatial and
  temporal cloaking.
\newblock In: Proc. of MobiSys, USENIX (2003)

\bibitem{Gedik:05:ICDCS}
Gedik, B., Liu, L.:
\newblock Location privacy in mobile systems: {A} personalized anonymization
  model.
\newblock In: Proc. of ICDCS, IEEE (2005)  620--629

\bibitem{Mokbel:06:VLDB}
Mokbel, M.F., Chow, C.Y., Aref, W.G.:
\newblock The new casper: Query processing for location services without
  compromising privacy.
\newblock In: Proc. of VLDB, ACM (2006)  763--774

\bibitem{Bamba:08:WWW}
Bamba, B., Liu, L., Pesti, P., Wang, T.:
\newblock Supporting anonymous location queries in mobile environments with
  privacygrid.
\newblock In: Proc. of WWW, ACM (2008)  237--246

\bibitem{Duckham:05:Pervasive}
Duckham, M., Kulik, L.:
\newblock A formal model of obfuscation and negotiation for location privacy.
\newblock In: Proc. of PERVASIVE. Volume 3468 of LNCS., Springer (2005)
  152--170

\bibitem{Xue:09:LoCa}
Xue, M., Kalnis, P., Pung, H.:
\newblock Location diversity: Enhanced privacy protection in location based
  services.
\newblock In: Proc. of LoCA. Volume 5561 of LNCS., Springer (2009)  70--87

\bibitem{Kido:05:ICDE}
Kido, H., Yanagisawa, Y., Satoh, T.:
\newblock Protection of location privacy using dummies for location-based
  services.
\newblock In: Proc. of ICDE Workshops. (2005)  1248

\bibitem{Shankar:09:UbiComp}
Shankar, P., Ganapathy, V., Iftode, L.:
\newblock {Privately querying location-based services with SybilQuery}.
\newblock In: Proc. of UbiComp, ACM (2009)  31--40

\bibitem{Herrmann:13:WPES}
Herrmann, M., Troncoso, C., Diaz, C., Preneel, B.:
\newblock Optimal sporadic location privacy preserving systems in presence of
  bandwidth constraints.
\newblock In: Proc. of WPES. (2013)

\bibitem{Shokri:14:TDSC}
Shokri, R., Theodorakopoulos, G., Papadimitratos, P., Kazemi, E., Hubaux, J.P.:
\newblock Hiding in the mobile crowd: Location privacy through collaboration.
\newblock In: Proc. of the TDSC, IEEE (2014)

\bibitem{Machanavajjhala:08:ICDE}
Machanavajjhala, A., Kifer, D., Abowd, J.M., Gehrke, J., Vilhuber, L.:
\newblock Privacy: Theory meets practice on the map.
\newblock In: Proc. of ICDE, IEEE (2008)  277--286

\bibitem{Ho:11:GIS}
Ho, S.S., Ruan, S.:
\newblock Differential privacy for location pattern mining.
\newblock In: Proc. of SPRINGL, ACM (2011)  17--24

\bibitem{Dewri:12:TMC}
Dewri, R.:
\newblock Local differential perturbations: Location privacy under approximate
  knowledge attackers.
\newblock IEEE Trans. on Mobile Computing \textbf{99}(PrePrints) (2012) ~1

\bibitem{Shokri:14:TechRep}
Shokri, R.:
\newblock Optimal user-centric data obfuscation.
\newblock Technical report, ETH Zurich (2014)
  \url{http://arxiv.org/abs/1402.3426}.

\bibitem{Dwork:12:ITCS}
Dwork, C., Hardt, M., Pitassi, T., Reingold, O., Zemel, R.S.:
\newblock Fairness through awareness.
\newblock In: Proc. of ITCS, ACM (2012)  214--226

\end{thebibliography}

%\break
%\newpage
%\vfill\eject

\appendix
% !TEX root = geoopt.tex

\section{Proofs}

\propprivacyimpl*
\begin{proof}
This proposition is a direct consequence of the property
\[
d_G(x,x') \leq \delta \dx(x,x') \quad \forall x, x' \in \calx
\]
and one of the results presented in \cite{Chatzikokolakis:13:PETS}, which states that if two metrics $\dx$ and $\dy$ are such that $\dx \leq \dy$ (point-wise), then \priv{\dx} implies \priv{\dy}.\qed
\end{proof}

\theoremap*
\begin{proof}
We know that 
\[
(KH)_{x\hat{x}} = \sum_{z\in \calx} k_{xz}h_{z\hat{x}}, \quad \forall x, \hat{x}\in\calx
\]
Since $K$ is \privadj{\edx}, we also know that
\[
k_{xz} \leq e^{\epsilon \dx(x,x')} k_{x'z}, \quad \forall x, x', z \in\calx
\]
Therefore, given $x, x'\in \calx$, it holds that for all $\hat{x}\in\calx$:
\begin{align*}
(KH)_{x\hat{x}} 	&= \sum_{z\in \calx} k_{xz}h_{z\hat{x}}\\
				&\leq \sum_{z\in \calx} e^{\epsilon \dx(x,x')} k_{x'z}h_{z\hat{x}}\\
				&= e^{\epsilon \dx(x,x')} \sum_{z\in \calx} k_{x'z}h_{z\hat{x}}\\
				&= e^{\epsilon \dx(x,x')} (KH)_{x'\hat{x}} 
\end{align*}
and therefore $KH$ is \privadj{\edx}.\qed
\end{proof}

\theodpopt*
\begin{proof}
Let $\da = \dq$. We recall from Section \ref{sec:rezas-mech} that for an arbitrary mechanism $M$, it holds that

\begin{align*}
\advarg{M}{\uprof}{\dq} 	&= \min_{H} \expdistarg{MH}{\uprof}{\dq}\\
					&= \min_{H} \sqlarg{MH}{\uprof}{\dq}
\end{align*}

which means that 

\[
\advarg{M}{\uprof}{\dq} \leq \sqlarg{M}{\uprof}{\dq}\tag{1}
\]

Let $K$ be a \dpopt{\uprof}{\dx}{\dq} mechanism. Suppose that

\[
\advarg{K}{\uprof}{\dq} < \sqlarg{K}{\uprof}{\dq}
\]

This means that there is a remapping $H$, other than the identity, such that 
%\break
%\vfill\eject
%\newpage 
\[
\sqlarg{KH}{\uprof}{\dq}\ < \sqlarg{K}{\uprof}{\dq}
\]

However, by Lemma \ref{theo:remap} we know that $KH$ is also \privadj{\dx}, and therefore, recalling Definition \ref{def:dx-opt}, $K$ would not be \dpopt{\uprof}{\dx}{\dq}, which is a contradiction.
Therefore, we can state that

\[
\advarg{K}{\uprof}{\dq} = \sqlarg{K}{\uprof}{\dq}\tag{2}
\]

Now, in order to see that $K$ is also \optpriv{\prior}{\dq}{\dq}{q}, with $q = \sqlarg{K}{\uprof}{\dq}$, let $K'$ be such that 

\[
\sqlarg{K'}{\uprof}{\dq}\ \leq \sqlarg{K}{\uprof}{\dq}\tag{3}
\]

According to Definition \ref{def:sql-opt} we need to prove that 

\[
\advarg{K'}{\uprof}{\dq} \leq \advarg{K}{\uprof}{\dq}
\]

And in fact we can see that

\begin{align*}
\advarg{K'}{\uprof}{\dq} 	&\leq \sqlarg{K'}{\uprof}{\dq}\tag{by (1)}\\
					&\leq \sqlarg{K}{\uprof}{\dq}\tag{by (3)}\\
					&= \advarg{K}{\uprof}{\dq}\tag{by (2)}
\end{align*}

which concludes our proof.\qed

\end{proof}
%\vfill\eject

\section{Dual form of the optimization problem}
\label{app:dual}

\smallskip

In this section we  show the dual form of the optimization problem presented in Section \ref{sec:mechanism-spanner}. 
%\newpage
% \vfill\eject
We recall that the original linear program is as follows:

\begin{align*}
	&\textbf{Minimize:}   &\quad \sum_{x, z \in \calx} \uprof_x k_{xz} \dq(x, z) \\
	\\
	&\textbf{Subject to:}\\
			&\quad k_{xz} \leq e^{\frac{\epsilon}{\delta} d_G(x,x')} k_{x'z} & z \in \calx, (x,x') \in E \tag{1}\\
			&\quad \sum_{x \in \calx} k_{xz} = 1&  x\in \calx \tag{2}\\
			&\quad k_{xz} \geq 0 & x, z \in \calx \\
\end{align*}

To obtain the dual form, we apply the standard technique of linear programming.  
\vfill\eject

First, for the dual program we need to consider one variable for each of the constraints in the original linear program that are not constraints on single variables. 
Therefore we have two sets of variables: 
\begin{itemize}
\item The variables of the form ${a_{xx'z}}$, with $z \in \calx, (x,x') \in E$, corresponding to the constraints in (1).
\item The variables of the form $b_x$, with $x \in \calx$, corresponding to the constraints in (2).
\end{itemize}

Again applying the standard technique,  we obtain the following  system of constraints and  objective function, that constitute the dual linear program:

%\begin{align*}
%	\textbf{Maximize:}   &\quad \sum_{x \in \calx} b_x \\
%	\\
%	\textbf{Subject to:}
%			&\quad b_x + \sum_{(x, x') \in E} ( e^{\edg(x, x')} a_{x'xz} - a_{xx'z} ) \leq \uprof_x d_{Q}(x,z) & x, z \in \calx \\
%			&\quad a_{xx'z} \geq 0 &  z \in \calx, (x,x') \in E \\
%\end{align*}

\begin{align*}
	&\textbf{Maximize:}  \quad \sum_{x \in \calx} b_x \\
	&\textbf{Subject to:}\\
			&\quad b_x + \sum_{(x, x') \in E} ( e^{\edg(x, x')} a_{x'xz} - a_{xx'z} ) \leq \uprof_x d_{Q}(x,z),\quad x, z \in \calx \\
			&\quad a_{xx'z} \geq 0, \quad  z \in \calx, (x,x') \in E \\
\end{align*}

\end{document}